\documentclass[10pt,journal,compsoc]{IEEEtran}
  \usepackage{cite}
\usepackage{graphicx}
\usepackage{subfigure}
\usepackage{epstopdf}
\usepackage{multicol}
\usepackage{algorithm}
\usepackage{algorithmic}
\usepackage{amsthm,amsmath,amsfonts}
\usepackage{mathrsfs}
\usepackage{courier}
\newtheorem{lemma}{\textbf{Lemma}}

\newtheorem{assumpt}{\textbf{Assumption}}

\newtheorem{propo}{\textbf{Proposition}}
\usepackage{bm}
\usepackage{subfigure}
\usepackage{color}
\usepackage{ragged2e} 
\usepackage{ragged2e}
\usepackage{booktabs}

\ifCLASSINFOpdf
\else
\fi
\begin{document}
%
\title{QoE-Driven Video Transmission: Energy-Efficient Multi-UAV Network Optimization}
%
%
%
%

\author{Kesong Wu, Xianbin Cao,~\IEEEmembership{Senior Member,~IEEE}, Peng Yang,~\IEEEmembership{Member,~IEEE}, Zongyang Yu,

Dapeng Oliver Wu,~\IEEEmembership{Fellow,~IEEE}, and
Tony Q. S. Quek,~\IEEEmembership{Fellow,~IEEE}

\thanks{Manuscript received 10 January 2023; revised 4 June 2023; accepted 21 July 2023. 
This work was supported in part by the National Natural Science Foundation of China under Grant 61827901 and Grant 62201024, and in part by Application Research of UAV Inspection System in Facilities and Bridges of Shuohuang Railway under Grant GJNY-19-90. This article was presented in part at the 2023 IEEE/CIC International Conference on Communications in China \cite{Wu2023ICCCChina}. (Corresponding author: Peng Yang and Xianbin Cao)

K. Wu, X. Cao, P. Yang, and Z. Yu are with the School of Electronic and Information Engineering, Beihang University, Beijing 100191, China.

D. O. Wu is with the Department of Computer Science, City University of Hong Kong, Kowloon, Hong Kong.

T. Q. S. Quek is with the Information Systems Technology and Design Pillar, Singapore University of Technology and Design, Singapore 487372, Singapore.
}
}

%
%

\markboth{Journal of \LaTeX\ Class Files, 2023}
{Shell \MakeLowercase{\textit{et al.}}: Bare Demo of IEEEtran.cls for Computer Society Journals}
%



\IEEEtitleabstractindextext{%
\justifying\let\raggedright\justifying
\begin{abstract}
This paper is concerned with the issue of improving video subscribers' quality of experience (QoE) by deploying a multi-unmanned aerial vehicle (UAV) network. Different from existing works, we characterize subscribers' QoE by video bitrates, latency, and frame freezing and propose to improve their QoE by energy-efficiently and dynamically optimizing the multi-UAV network in terms of serving UAV selection, UAV trajectory, and UAV transmit power. The dynamic multi-UAV network optimization problem is formulated as a challenging sequential-decision problem with the goal of maximizing subscribers' QoE while minimizing the total network power consumption, subject to some physical resource constraints. We propose a novel network optimization algorithm to solve this challenging problem, in which a Lyapunov technique is first explored to decompose the sequential-decision problem into several repeatedly optimized sub-problems to avoid the curse of dimensionality. To solve the sub-problems, iterative and approximate optimization mechanisms with provable performance guarantees are then developed. Finally, we design extensive simulations to verify the effectiveness of the proposed algorithm. Simulation results show that the proposed algorithm can effectively improve the QoE of subscribers and is 66.75\% more energy-efficient than benchmarks. 
\end{abstract}

\begin{IEEEkeywords}
Video transmission, QoE, multi-UAV network, network optimization, energy efficiency
\end{IEEEkeywords}}

\maketitle

\IEEEdisplaynontitleabstractindextext

%
\IEEEpeerreviewmaketitle

\IEEEraisesectionheading{\section{Introduction}\label{sec:introduction}}
%
%
%
%
\IEEEPARstart{C}{ompared} to traditional terrestrial wireless networks, unmanned aerial vehicle (UAV) networks can provide a more rapid and flexible networking capability and have stronger resistance to natural disasters. In the case that terrestrial base stations (BSs) in hot spots are insufficient to handle burst flows or terrestrial communication infrastructure in remote areas cannot provide effective coverage, UAV networks with UAVs acting as aerial access points or relay nodes have been considered as a significant complement to terrestrial wireless networks
\cite{DBLP:journals/comsur/GeraciGALMCCRR22, DBLP:journals/comsur/BaltaciDOACS21, DBLP:journals/comsur/HassijaCAGLNYG21,Cao2018Airborne}. 
UAVs can actively establish line-of-sight (LoS) propagation links among them and terrestrial subscribers such that the network transmission performance can be significantly enhanced. 
Overall, UAV networks are widely employed in emergency communications and some dynamic traffic demand scenarios due to the UAV endurance improvement, the cost reduction, the advantages of flexible deployment, and the rapid recovery of communication services \cite{DBLP:journals/tnse/ChangT22,DBLP:journals/tnse/ShehzadAHJ21,DBLP:journals/tnse/XuOD20}. 
Besides, UAV networks have been considered as significant components of the air-space-ground integrated information network, which is recognized as a crucial development direction for 6G \cite{Tang2020Future}.

Consequently, UAV networks have recently received extensive attention from both academia and industry, and a large number of studies on UAV-assisted communications in terrestrial networks have been proposed during the past several years \cite{DBLP:journals/tnse/JiangMLHCH22,DBLP:journals/tnse/CaiYCCZYD22,DBLP:journals/tnse/WangJTXS21,DBLP:journals/tnse/ShimadaKK21,DBLP:journals/tnse/HanyuKK20,DBLP:journals/tnse/LiuZQGM22,DBLP:journals/tnse/NikoorooB22}.
UAV networks can provide various types of services, among which video streaming media has been the preferred and dominant way of presenting information, particularly in the field of emergency rescue \cite{Burhanuddin2022QoE}. Additionally, with the development of broadband mobile networks and the proliferation of smart mobile devices, video streaming accounts for an increasing share of network traffic. Predictably, the share will rise further in the coming 6G era.
Therefore, the research on video transmission in UAV networks is not only important for the rapid popularization of UAV networks in emergency communications, but also has theoretical guidance and practical reference significance for promoting the application and development of 6G.

Video transmission is characterized by the transmission of a large amount of information and the sequential playback, which leads to high requirements on the network for high throughput and low latency, especially in light of the trend toward high definition. The quality of video streaming transmission is closely related to the performance of networks. However, the uncertainty of the link quality of UAV networks caused by dynamic network topology, complex signal interference, and time-varying channel fading pose a great challenge
to video transmission research in UAV networks
\cite{Zhan2021Joint}. 

\subsection{State of the Art}
Many recent studies on UAV network optimization for effective video transmission have been conducted.
For example, in \cite{Ros2017Relay}, the authors introduced a relay placement mechanism to find the ideal locations for UAVs acting as relay nodes, and thus avoided the effects of UAV movements on the video dissemination.
The authors in \cite{Zhu2018Transmission} investigated
the joint optimization of UAV trajectory and transmission rate allocation for reliable video streaming delivery in UAV networks.
In \cite{Jiang2019Multimedia}, the authors considered the joint optimization of the UAV deployment and the content placement of a cache-enabled UAV for the maximization of link throughput when delivering multimedia data.
In conclusion, the above works of transmitting video streams using UAV networks mainly focused on improving the link throughput by optimizing the UAVs’ locations, trajectories, and network resources without considering the proprietary characteristics of video streaming. 

QoE, which can effectively capture the unique characteristics of video streaming, is an important measure of the performance of video transmission. Recently, many researchers have paid attention to the QoE-driven UAV network optimization. 
For instance, in \cite{Hu2019Optimization}, the authors created a QoE utility model and studied an average QoE maximization problem for video transmission in a UAV relay system by optimizing the system bandwidth and power allocation. 
However, the work assumed that QoE depended only on the video transmission rate.
In \cite{Tang2021QoE-Driven}, the authors proposed a QoE-driven dynamic pseudo analog wireless video broadcasting scheme, and the goal was to maximize the peak signal-to-noise ratio (PSNR) of a subscriber’s reconstructed videos by jointly optimizing the UAV's transmit power and trajectory. 
Nevertheless, the evaluation metric (exactly, PSNR), which does not consider the overall video sequences, is typically used to quantify the distortion degree of video images. 
For video streaming, the key performance indicators include video bitrates, frame freezing, and latency.

To this end, the authors in \cite{Burhanuddin2022QoE} studied a scenario of the UAV-assisted live video streaming transmission and modeled the QoE using latency, video resolution, and smoothness. The authors in \cite{Bera2020QoE} presented a QoE-guided content delivery framework for providing subscribers with on-demand content in a multi-UAV network and aimed to maximize the QoE by improving the average end-to-end delay for each subscriber and the average throughput of each UAV. 
The work in \cite{Chen2020Optimal} proposed to model QoE by the bitrate together with the frozen time of videos and formulated a problem of optimizing the system bandwidth along with UAVs' transmit power to maximize the total long-term QoE.
Compared with the previous works \cite{Hu2019Optimization,Tang2021QoE-Driven}, the studies in \cite{Burhanuddin2022QoE,Bera2020QoE, Chen2020Optimal} exploited more indicators in constructing a QoE model. However, the energy efficiency of the UAV networks were not investigated in \cite{Burhanuddin2022QoE,Hu2019Optimization,Tang2021QoE-Driven,Bera2020QoE,Chen2020Optimal}. Compared to traditional BSs, UAVs are sharply energy-sensitive, and then it is essential to optimize the energy consumption of the UAV networks. 

Therefore, the work in \cite{Zhang2021Energy-Efficient} jointly optimized video level selection and power allocation to maximize the energy efficiency of a UAV network, which was defined as the ratio of video bitrate to UAV power consumption. 
The authors in \cite{Zeng2020Resource} jointly optimize subscriber communication scheduling, UAV trajectory, transmit
power, and bandwidth allocation to maximize the energy efficiency of the UAV network and
satisfy subscribers' QoE requirements.
Besides, the issue of energy-efficient trajectory optimization for aerial video surveillance under QoS constraints was investigated in \cite{Zhan2022Energy-Efficient}. 
Although the issue of energy-efficient UAV video streaming was investigated in \cite{Zhang2021Energy-Efficient,Zeng2020Resource,Zhan2022Energy-Efficient}, they just discussed the case of deploying a single UAV, which had restricted coverage, limited communication capacity, and stringent energy limitation. The availability of a single UAV network is also not guaranteed during the entire mission.
In contrast, multiple or a swarm of UAVs can serve subscribers collaboratively to achieve communications of higher throughput and lower latency. The availability and efficiency of communications can also be improved by deploying a multi-UAV network.
Nevertheless, the scheduling of communication resources of a multi-UAV network is much more challenging than that of a single UAV network \cite{Yang2022Fresh}.
The problem of deploying a multi-UAV network to provide efficient services for subscribers by jointly optimizing the association among UAVs and subscribers, UAVs' trajectories, and UAVs' transmit power was investigated in \cite{DBLP:journals/tcom/XuZLYXT22}. However, this work did not look into the particularities of video transmission.
The authors in \cite{He2019QoE-Oriented} designed an intelligent and distributed allocation mechanism to allocate uplink bandwidth for multi-UAV video streaming. 
Nevertheless, they just aimed to resolve the insufficiency of wireless channel resources when a cluster of UAVs executed the video shooting and uploading mission, nor did they consider the interference between UAVs.
\subsection{Motivations and Contributions}
From the above works \cite{Burhanuddin2022QoE,Zhan2021Joint,Ros2017Relay,Zhu2018Transmission,Jiang2019Multimedia,Hu2019Optimization,Tang2021QoE-Driven,Bera2020QoE,Chen2020Optimal,Zhang2021Energy-Efficient,Zeng2020Resource,Zhan2022Energy-Efficient}, we observe that optimizing the resource allocation to improve the video transmission performance of UAV networks has become an important and popular research topic. 
However, the issue of improving subscribers' QoE, including video bitrates, latency, and frame freezing by controlling a multi-UAV network is not investigated. 
This paper proposes to completely control multiple UAVs in terms of their trajectories, transmit power, and serving UAV selection to proactively alter downlink video transmission channels. In this way, it is desired to achieve a significant improvement in the video transmission performance of the multi-UAV network, including higher video bitrates, lower latency, and alleviated frame freezing.
The main contributions of this paper are summarized as follows:

    1) We mathematically model the QoE of subscribers based on the comprehensive analysis of the characteristics of the multi-UAV network and video transmission. In particular, a novel video streaming utility model adaptively matching the video bitrate and the playback bitrate requirements of subscribers is designed. Besides, the time-varying startup and rebuffering latency is investigated and modeled, and a constraint on the minimum time-averaged achievable bitrate is enforced to alleviate frame freezing. 
    
    2) A dynamic multi-UAV network is desired to be deployed to deliver video streams to subscribers, considering the restricted UAV communication coverage. To this end, we formulate the problem of dynamically deploying a multi-UAV network to deliver video streams to subscribers as a sequential-decision problem. The goal of this problem is to maximize the QoE of subscribers and minimize the total power consumption of the multi-UAV network through the joint optimization and control of UAV transmit power, serving UAV selection, and UAV trajectory. The formulated problem is confirmed to be NP-hard or non-tractable, and solving the sequential-decision problem as a whole encounters the curse of dimensionality. 
    
    3) Considering the advantages of the Lyapunov technique in tackling sequential-decision optimization problems, we develop a Lyapunov-based network optimization algorithm with provable performance guarantees to solve the optimization problem. First, we decompose the optimization problem into several repeatedly optimized sub-problems to avoid the curse of dimensionality. Nevertheless, the sub-problems are still difficult to solve
    as they are mixed-integer and non-convex. We then introduce the key idea of iterative optimization to tackle the mixed-integer issue and explore Taylor expansions to identify the convex approximation of the non-convex feasible regions of the sub-problems. 
    
    4) Finally, we compare the proposed algorithm with various benchmarks to verify its effectiveness and design extensive simulations to discuss the impact of diverse parameters on the performance of the algorithm. Simulation results show that the proposed algorithm can effectively improve the QoE of subscribers and is 66.75\% more energy-efficient than benchmarks.
\begin{table}[!t]
\renewcommand{\arraystretch}{1.2}
\caption{{SUMMARY OF IMPORTANT ACRONYMS}}
\label{table_1}
\newcommand{\tabincell}[2]{\begin{tabular}{@{}#1@{}}#2\end{tabular}}
\centering
\begin{tabular}{ll}
\toprule
{\textbf{Acronyms}} & {\textbf{Meaning}} \\
\midrule
{2D} & {Two-dimensional} \\
{6G} &    {Sixth-Generation mobile networks} \\
{AtG} &   {Air-to-ground} \\
{BSs} &   {Base stations} \\
{CUMTP} & {\tabincell{l}{Circular UAV trajectory with the maximum transmit \\ power}} \\
{CUTR} &  {Circular UAV TRajectory} \\
{LoS}  &  {Line-of-sight} \\
{NNAS} &  {Nearest Neighbor ASsociation} \\
{NLoS} &  {Non-line-of-sight} \\
{NP-hard} & {Non-deterministic Polynomial hard} \\
{PSNR} &  {Peak signal-to-noise ratio} \\
{QoE}  & {Quality of experience} \\
{RHS}  &  {Right-hand-side} \\
{SUMTP} & {Stationary UAV with maximum transmit power} \\
{SUDE} &  {Stationary UAV DEployment} \\
{UAV}  &  {Unmanned aerial vehicle} \\
{VSS} &   {Video source station} \\
\bottomrule
\end{tabular}
\end{table}

\begin{table*}[!t]
	\renewcommand{\arraystretch}{1.2}
	\caption{{LIST OF SOME IMPORTANT NOTATIONS}}
	\label{table_2}
	\newcommand{\tabincell}[2]{\begin{tabular}{@{}#1@{}}#2\end{tabular}}
	\centering
	\begin{tabular}{|l||l|l||l|}
		\hline
		{Notation} & {Description} & {Notation} & {Description} \\\hline
		{${\mathcal K}$, ${\mathcal I}$} & {Set of UAVs and set of subscribers} & {$\bm{q_k}\left(t\right)$} & {Horizontal location of UAV $k$ at time slot $t$} \\ \cline{1-4}
		{$[{x_k^{(u)}}\left(t\right),{y_k^{(u)}}\left(t\right)]$} & {\tabincell{l}{2D horizontal coordinate of \\ UAV $k$ at  time slot $t$}} & {${D_{ik}}\left( t \right)$} & {\tabincell{l}{Distance between UAV $k$ and subscriber $i$ at time slot $t$}} \\ \cline{1-4}
		{$\bm{s_i}\left( t \right)$} & {Location of subscriber $i$ at time slot $t$}
		& {${h_{ik}}\left( t \right)$} & {Channel gain from UAV $k$ to subscriber $i$ at time slot $t$} \\ \cline{1-4}
		{$sin{r_{ik}}\left( t \right)$} & {\tabincell{l}{Signal-to-interference-plus-noise ratio \\ experienced by subscriber $i$ at time slot $t$}} & {${I_{ik}}\left( t \right)$} & {\tabincell{l}{Interference signal received by subscriber $i$ from other \\ UAVs except for serving UAV $k$ at time slot $t$}} \\ \cline{1-4}
		{${c_{ik}}\left( t \right)$} & {\tabincell{l}{Indicator of whether subscriber $i$ can select UAV \\ $k$ as its serving UAV}} & {${r_i}\left( t \right)$} & {Achievable bitrate of subscriber $i$ at time slot $t$} \\ \cline{1-4}
		{${{\bar r}_i}\left( t \right)$} & {\tabincell{l}{Time-averaged achievable bitrate of subscriber $i$ \\ during the first $t$ time slots
		}} & {${p_k}\left( t \right)$} & {Instantaneous transmit power of UAV $k$ at time slot $t$} \\ \cline{1-4}
		{${{\bar p}_k}\left( t \right)$} & {\tabincell{l}{Time-averaged transmit power of UAV $k$ during \\ the first $t$ time slots}} & {$d_{\min }$} & {Minimum safety distance between any two UAVs} \\ \cline{1-4}
        {$p_k^{tot}\left( t \right)$} & {\tabincell{l}{Total communication power consumption of \\ UAV $k$ at time slot $t$}} & {${s_{\max}}$} & {UAV’s maximum flight distance in a time slot} \\ \cline{1-4} 
		{${R_i}$} & {Required playback bitrate of subscriber $i$} & {$r_i^{th}$} & {\tabincell{l}{Threshold of time-averaged achievable bitrate of \\  subscriber $i$}} \\ \cline{1-4}
		{${\delta t}$} & {Duration of a time slot} & {$d\left( t \right)$} & {Total latency for all subscribers ${i \in {\mathcal I}}$ at time slot $t$} \\ \cline{1-4}
	\end{tabular}
\end{table*}

\begin{figure}[!t]
\centering
\includegraphics[width=3.2 in, height = 1.8 in]{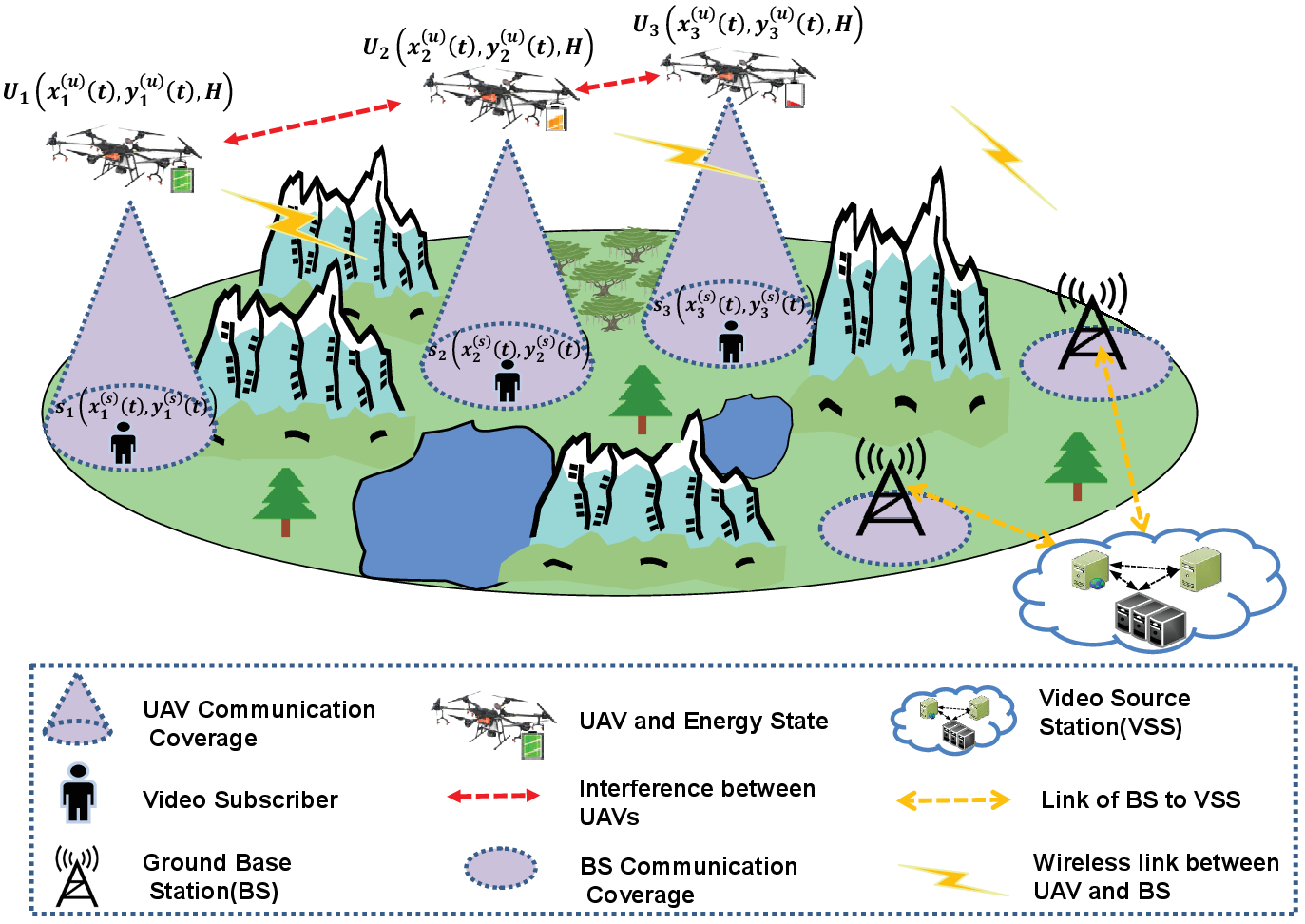}
\caption{Video transmission with a multi-UAV network.}
\label{fig_uav_video_scenario}
\end{figure}
\section{system model and Problem formulation}
\subsection{System Model}
This paper considers an application scenario for providing high-quality video streaming services for a number of ground subscribers by deploying a multi-UAV network, as shown in Fig. \ref{fig_uav_video_scenario}. 
In the considered scenario, there is a ground content provider acting as a video source station (VSS), several ground BSs, a multi-UAV network consisting of $N$ flying UAVs, the set of which is denoted by ${\mathcal K} =\{1, 2, 3, \ldots, N\}$, and $M$ subscribers, the set of which is ${\mathcal I}=\{1, 2, 3, \ldots, M\}$. The content provider has the capabilities of video collection, video transcoding, as well as video streaming push and can provide adaptive streaming services for subscribers. 
In this scenario, we investigate the issue that subscribers cannot directly access the VSS through BSs, mainly because of severe signal blocking. Considering that UAVs can establish LoS links with BSs, we propose to dynamically deploy a multi-UAV network with UAVs acting as flying relays to forward video streams from BSs to subscribers. 

To facilitate the mathematical model of the video transmission task, the time domain is discretized into a sequence of time slots, denoted by $t = \{1, 2, \ldots\}$.
As the communication coverage of UAVs is limited, the trajectories of UAVs will be continuously optimized to provide transmission services for subscribers fairly. We denote the time-varying two-dimensional (2D) horizontal coordinate of UAV $k$ at time slot $t$ by $\bm{q_k}\left(t\right) = [{x_k^{(u)}}\left(t\right),{y_k^{(u)}}\left(t\right)]^{\rm T},1 \le k \le N$. Denote the location of subscriber $i$ at time slot $t$ by $\bm{s_i}\left( t \right) = {\left[ {{x_i^{(s)}}\left( t \right),{y_i^{(s)}}\left( t \right)} \right]^{\rm T}},1 \le i \le M$. All UAVs are deployed at the same altitude, $H$, and the altitudes of all subscribers are negligible compared to $H$. Thus, the distance between UAV $k$ and subscriber $i$ at time slot $t$ can be given by ${D_{ik}}\left( t \right) = \sqrt {{H^2} + {{\left\| {\bm{q_k}\left( t \right) - \bm{s_i}\left( t \right)} \right\|}^2}}$.
Considering the stringent QoE requirements of subscribers, it is critical to provide video streaming services via air-to-ground (AtG) LoS links. The reliability and the throughput of non-line-of-sight (NLoS) propagation links cannot be guaranteed for video streaming services. High energy consumption will also be required to compensate for attenuation of NLoS propagation \cite{Ji2020Joint,Ji2021Joint}. 
Besides, according to the research results in \cite{Hourani2014Optimal}, the probability of LoS propagation can exceed $90\%$ when the elevation angle between a UAV and a ground subscriber is greater than a threshold $\theta$, which is related to the propagation environment, such as rural area, urban area, and dense urban. 
Therefore, the following condition can be held in order to approximate the establishment of a LoS link between a UAV and a ground subscriber.

\begin{equation}\label{eq:tan_theta}
\left\| {\bm{q_k}\left( t \right) - \bm{s_i}\left( t \right)} \right\| \le H{\tan ^{ - 1}}\theta ,\forall i,k,t
\end{equation}

The channel gain of an AtG LoS link can be calculated by the free-space path loss model{\cite{Zhan2021Joint}}, ${h_{ik}}(t) = \frac{{g_{ik}^Tg_{ik}^R{c^2}}}{{{{\left( {4\pi f_c{D_{ik}}\left( t \right)/{D_0}} \right)}^2}}}$
where ${h_{ik}}(t)$ denotes the channel gain from UAV $k$ to subscriber $i$ at time slot $t$, $g_{ik}^T$ and $g_{ik}^R$ respectively represent transmitting and receiving antenna gains, $c$ is the speed of light, $f_c$ is the carrier frequency, and ${D_0}$ is a far-field reference distance.
Let ${\omega _{ik}} = \frac{{g_{ik}^Tg_{ik}^R{c^2}D_0^2}}{{{{\left( {4\pi f_c} \right)}^2}}}$, ${h_{ik}}(t)$ can be rewritten as 
\begin{equation}\label{eq:h_ik_t_new}
{h_{ik}}(t) = \frac{{{\omega _{ik}}}}{{{H^2} + {{\left\| {\bm{q_k}\left( t \right) - \bm{s_i}\left( t \right)} \right\|}^2}}}
\end{equation}

Owing to the movement of UAVs and their limited coverage, a subscriber's serving UAV will be time-varying. We denote the serving UAV selection set at time slot $t$ as ${\mathcal C}(t)$. For any ${c_{ik}}\left( t \right) \in {\mathcal C}(t)$, ${c_{ik}}\left( t \right) = 1$ indicates that subscriber $i$ can select UAV $k$ as its serving UAV at time slot $t$; otherwise, ${c_{ik}}\left( t \right) = 0$. 
We assume that at time slot $t$, a subscriber can select at most one UAV as its serving UAV, and a UAV is allowed to deliver video streams to at most one subscriber. Mathematically, we have
\begin{equation}\label{eq:association_value}
0 \le \sum\nolimits_{k \in {\mathcal K}} {{c_{ik}}\left( t \right)}  \le 1,0 \le \sum\nolimits_{i \in {\mathcal I}} {{c_{ik}}\left( t \right)}  \le 1
\end{equation}

As all UAVs share the frequency resource in the considered scenario, a subscriber $i$ will receive its intended signal from UAV $k$ and interference caused by other UAVs at each time slot. For subscriber $i$, the strength of its received interference can be calculated by ${I_{ik}}\left( t \right) = \sum\limits_{j \in {\mathcal K}\backslash \left\{ k \right\}} {{p_j}\left( t \right){h_{ij}}\left( t \right)} $ where ${p_j}\left( t \right)$ denotes the instantaneous transmit power of UAV $j$ at time slot $t$. 
Then, the signal-to-interference-plus-noise ratio (SINR) experienced by subscriber $i$ at time slot $t$ can be given by $sin{r_{ik}}\left( t \right) = \frac{{{p_k}\left( t \right){h_{ik}}\left( t \right)}}{{{\sigma ^2} + {I_{ik}}\left( t \right)}}$, where ${\sigma ^2}$ is the noise power.
According to Shannon’s channel capacity formula, we can calculate the achievable bitrate of subscriber $i$ at time slot $t$ by
\begin{equation}\label{eq:achievable_rate_subscriber_i}
{r_i}\left( t \right) = \sum\nolimits_{k \in {\mathcal K}} {{c_{ik}}\left( t \right){{\log }_2}\left( {1 + sin{r_{ik}}\left( t \right)} \right)}
\end{equation}

For subscriber $i$, its time-averaged achievable bitrate during the first $t$ time slots can be given by ${{\bar r}_i}\left( t \right) = \frac{1}{t}\sum\nolimits_{\tau  = 1}^t {{r_i}\left( \tau  \right)}$.

It is worth noting that the onboard energy of UAVs is limited and needs to be efficiently utilized. In this paper, we mainly consider the communication power consumption of UAVs (exactly, for forwarding video streams) and model the total communication power consumption of UAV $k$ at time slot $t$ as {\cite{Ji2020Joint}}
\begin{equation}\label{eq:total_power_at_t}
p_k^{tot}\left( t \right) = {p_k}\left( t \right) + p_k^c
\end{equation}
where $p_k^c$ is the onboard circuit power consumption of UAV $k$. The time-averaged transmit power of UAV $k$ in the first $t$ time slots can be computed by ${{\bar p}_k}\left( t \right) = \frac{1}{t}\sum\limits_{\tau  = 1}^t {{p_k}\left( \tau  \right)} $.
Accordingly, the time-averaged total communication power consumption of UAV $k$ in the first $t$ time slots can be written as $\bar p_k^{tot}\left( t \right) = {{\bar p}_k}\left( t \right) + p_k^c$.
The upper-bounded constraints of $p_k^{tot}\left( t \right)$ and $\bar p_k^{tot}\left( t \right)$ can be written as $p_k^{tot}\left( t \right) \le {{\hat p}_k}$ and $\bar p_k^{tot}\left( t \right) \le {{\tilde p}_k},\forall k,t$, where ${{\hat p}_k}$ and ${{\tilde p}_k}$ denote the maximum instantaneous total power consumption and the maximum time-averaged total power consumption of UAV $k$, respectively\cite{Yang2021Proactive}.

At any time slot $t$, to avoid the UAV collision, the distance between any two UAVs must be greater than a value {\cite{Yang2022Fresh},\cite{Wu2018Joint}}, i.e., 
\begin{equation}\label{eq:minimum_safety_distance}
{\left\| {\bm{q_k} \left( t \right) - \bm{q_j}\left( t \right)} \right\|^2} \ge d_{\min }^2,\forall k,k \ne j,t
\end{equation}
where ${d_{\min }}$ denotes the minimum safety distance.

In addition, limited by the flight speed, the moving distance of a UAV in a time slot is constrained{\cite{Yang2022Fresh},\cite{Wu2018Joint}}, i.e., 
\begin{equation}\label{eq:maximum_flight_speed}
{\left\| {\bm{q_k}\left( t \right) - \bm{q_k}\left( {t - 1} \right)} \right\|^2} \le s_{\max }^2,\forall k,t
\end{equation}
where ${s_{\max }}$ denotes the UAV’s maximum flight distance in a time slot.

\subsection{QoE Model}
QoE has become an acknowledged performance evaluation standard in video streaming services. The accurate analysis of the key factors influencing QoE will be beneficial for enhancing subscriber satisfaction and improving the utilization of network resources.
In the research field of video streaming, high resolution, fluency, low latency, and smoothness of video streaming are four universal factors affecting QoE. These influencing factors, however, are challenging to be optimized simultaneously, especially in a time-varying multi-UAV network. 
Meanwhile, the dominant influencing factors are diverse for various network services. For example, some applications are sensitive to latency, while others are greatly affected by packet loss ratio. The smoothness of video streaming is closely related to subscribers' switching strategies towards video streaming of diverse quality and cannot be directly optimized by controlling the multi-UAV network. Considering the time-varying characteristics of the multi-UAV network and the continuous playback requirements of video streaming, we create a novel QoE model incorporating the video bitrate, the latency, and the frame freezing in this paper. 

The decoding capabilities and screen sizes of subscriber terminals may be different, and their playback bitrate requirements for the same video streaming may be diverse. In this case, implementing adaptive transmission and playback of video streaming is an effective way of improving QoE. Based on the throughput of the multi-UAV network and the receiving capabilities of subscribers, we design an adaptive video streaming utility model that implements the adaptive matching between subscribers' required playback bitrates and video bitrates. For a subscriber, its available video bitrate is upper-bounded by its achievable bitrate; otherwise, packet loss or frame freezing will occur. 
Therefore, the video streaming utility model $\phi \left({\bm{\bar r}\left( t \right)} \right)$ for all subscribers $i \in {\mathcal I}$, which also represents the profit of the multi-UAV network, can be written as
\begin{equation}\label{eq:utility_model}
\phi \left( {\bm{\bar r}\left( t \right)} \right) = \alpha \sum\limits_{i = 1}^M {{{\log }_2}\beta (1 + \frac{{B{{\bar r}_i}\left( t \right)}}{{{R_i}}})}
\end{equation}
where $ {\bm{\bar r}\left( t \right)} = \left( {{{\bar r}_1}\left( t \right), \ldots ,{{\bar r}_M}\left( t \right)} \right)$, ${{R_i}}$ represents the required playback bitrate of subscriber $i$, {$B$ represents the total bandwidth,} and $\alpha $, $\beta$ are both positive values that are different for various types of applications. 

Latency, especially the startup latency and the rebuffering latency in video streaming, is a crucial factor affecting QoE. The startup latency and the rebuffering latency are closely related to subscribers' achievable bitrates. The greater the achievable bitrate, the shorter the latency to receive a sufficient amount of video data. Latency is also greatly affected by a subscriber's selection of serving UAV in the considered application scenario.
Hence, the startup and rebuffering latency (briefly, latency) model ${d_i}\left( t \right)$ for any subscriber $i \in {\mathcal I}$ at time slot $t$ is given by
\begin{equation}\label{eq:latency_subscriber_i}
 {{d_i}\left( t \right) = \sum\limits_{k \in {\mathcal K}} {\frac{{L{c_{ik}}\left( t \right)}}{{B{{\log }_2}\left( {1 + sin{r_{ik}}\left( t \right)} \right)}}}  + (1 - \sum\limits_{k \in {\mathcal K}} {{c_{ik}}\left( t \right)} )\delta t}
\end{equation}
where ${\delta t}$ denotes the duration of a time slot, and $L$ is the length of transmitted video data in a time slot. {According to (\ref{eq:latency_subscriber_i}), we will obtain the latency ${d_i}\left( t \right) = \frac{L}{{B{{\log }_2}\left( {1 + sin{r_{ik}}\left( t \right)} \right)}}$, when subscriber $i$ can select a certain UAV $k$ as its serving UAV at time slot $t$; otherwise, the latency ${d_i}\left( t \right)$ is $\delta t$, namely the interval of one time slot.}
$d\left( t \right) = \sum\nolimits_i {{d_i}\left( t \right)}$ denotes the total latency of all subscribers $i$ at time slot $t$. 

The exhaustion of resources in a subscriber's playback buffer will lead to annoying frame freezing. When the playback bitrate remains unchanged, either the decrease of link throughput or the link interruption can result in the exhaustion of buffer resources. To tackle this issue, it is essential to constrain the time-averaged achievable bitrate ${{\bar r}_i}\left( t \right)$ by
\begin{equation}\label{eq:throughput_constrain}
{{\bar r}_i}\left( t \right) \ge {r_i^{th}},\forall i
\end{equation}
where, $r_i^{th}$ denotes the threshold of time-averaged achievable bitrate of subscriber $i$. In order to meet this constraint, the optimization of UAV trajectories and network resources is desired and should be investigated.

\subsection{Problem Formulation}
Considering that UAVs are energy-sensitive, our goal is to maximize subscribers' QoE and minimize the total power consumption of the multi-UAV network. To achieve this goal, we propose to jointly optimize the serving UAV selection ${\mathcal C}(t)$, the UAVs' trajectories ${\mathcal Q}(t)$, and the UAVs' transmit power ${\mathcal P}(t)$.
According to the above analysis, the optimization problem can be mathematically formulated as follows
\begin{subequations}\label{eq:original_problem}
\begin{alignat}{2}
&{\mathop {{\rm{Maximize}}}\limits_{{\mathcal C}\left( t \right),{\mathcal P}\left( t \right),{\mathcal Q}\left( t \right)} \mathop {\lim \inf }\limits_{t \to \infty } \phi \left( {\bm{\bar r}\left( t \right)} \right) - {\rho _2}\sum\limits_{k \in {\mathcal K}} {\bar p_k^{tot}\left( t \right)}  - }\notag\\
&{{\rho _1}\sum\limits_i ( \sum\limits_{k \in {\mathcal K}} {\frac{{L{c_{ik}}\left( t \right)}}{{B{{\log }_2}\left( {1 + sin{r_{ik}}\left( t \right)} \right)}}}  + (1 - \sum\limits_{k \in {\mathcal K}} {{c_{ik}}\left( t \right)} )\delta t)}\\
&{\rm s.t.}\mathop {\lim \inf}\limits_{t \to \infty } {{\bar r}_i}\left( t \right) \ge {r_i^{th}},\forall i\\
&\mathop {\lim \sup}\limits_{t \to \infty }  \bar p_k^{tot}\left( t \right) \le {{\tilde p}_k},\forall k \allowdisplaybreaks[4] \\
&p_k^{tot}\left( t \right) \le {{\hat p}_k},\forall k,t\\
&{p_k}\left( t \right) \ge p_k^{\min },\forall k,t\\
&{c_{ik}}\left( t \right) \in \{ 0,1\} ,\forall i,k,t\\
&\rm{constraints\text{ }(\ref{eq:tan_theta}),\text{ } (\ref{eq:association_value}),\text{ } (\ref{eq:minimum_safety_distance}),\text{ } (\ref{eq:maximum_flight_speed})\text{ } are \text{ }satisfied}
\end{alignat}
\end{subequations}
where ${{\rho _1}}$ and ${{\rho _2}}$ are both positive values that reflect the trade-off among the network profit, the total latency, and the total power consumption. One can choose ${{\rho _1}}$ and ${{\rho _2}}$ based on preferences of network operators and subscribers and the Pareto frontier of the formulated multi-objective optimization problem, as depicted in Fig. \ref{fig_Pareto_frontier}.
\begin{figure}[!t]
\centering
\includegraphics[width=3.2 in, height = 1.6 in]{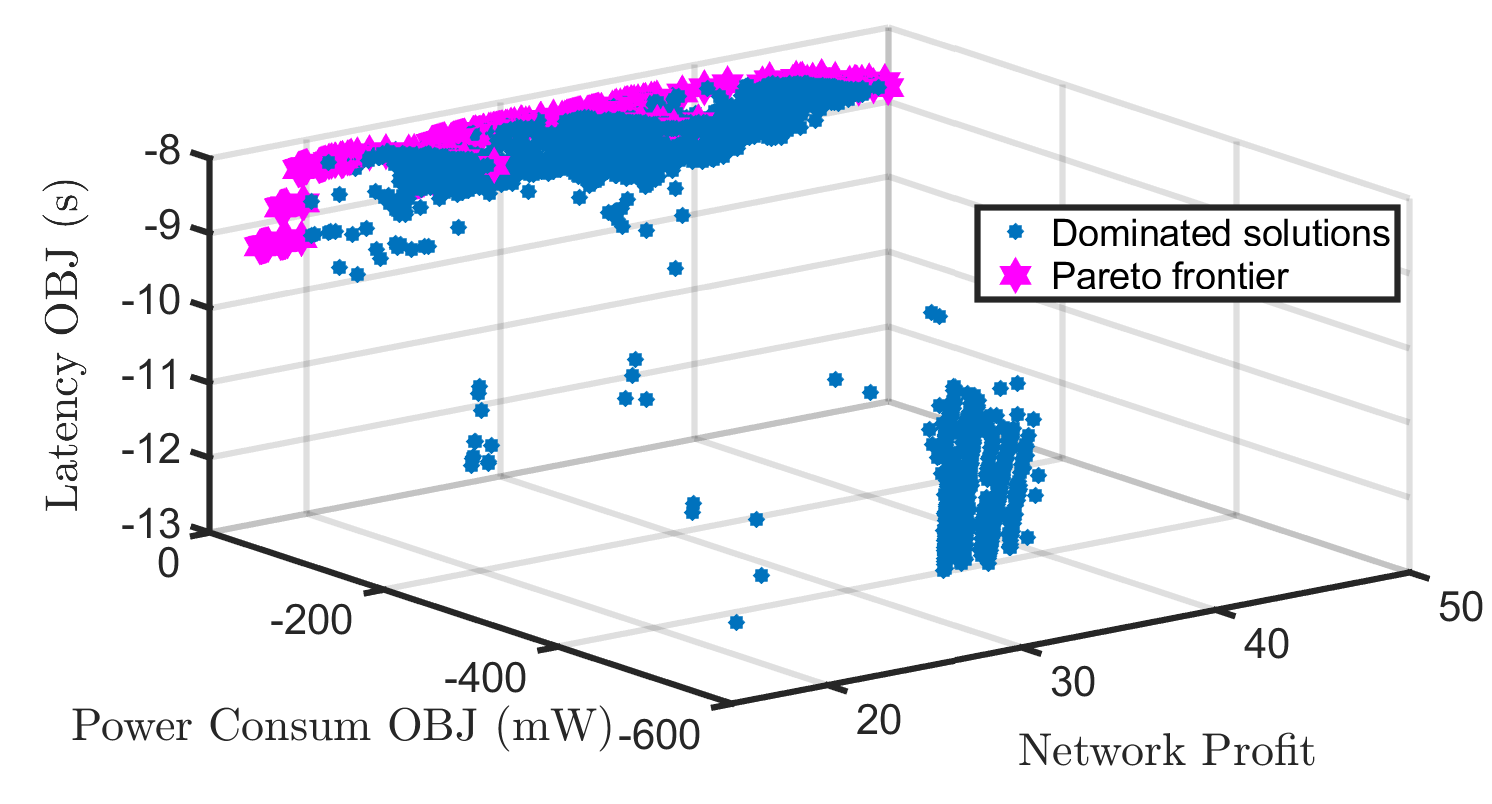}
\caption{Pareto frontier and dominated solutions of (\ref{eq:original_problem}).}
\label{fig_Pareto_frontier}
\end{figure}

The formulated problem is confirmed to be a sequential-decision optimization problem, the solution of which is highly challenging. 
The video streaming utility model in the objective function is a logarithm function of time-averaged achievable bitrates of subscribers during the first $t$ time slots, which not only makes the problem highly time-coupled but also greatly increases the computational complexity of solving the problem. For example, the number of time-varying decision variables of the problem will grow exponentially with the increase of time slot $t$. As a result, to solve this problem directly by exploring some conventional optimization algorithms is unacceptable in terms of computational complexity.
Besides, this optimization problem includes the summation term of 2-norm, complex fractional terms, continuous and integer variables, logarithmic-quadratic functions, and non-convex constraints. Thus, the formulated problem is also a mixed-integer non-convex optimization problem that may be NP-hard or even undecidable \cite{Belotti2013Mixed-integer}.

To alleviate this highly challenging problem, we first employ the Jensen's inequality to decouple the time-coupled objective function. Next, considering the advantages of Lyapunov optimization approach in tackling time-averaged problems, we decompose the sequential-decision problem into multiple repeatedly optimization sub-problems at different time slots using the Lyapunov approach. In this way, the curse of dimensionality can be effectively tackled, and the computational complexity of solving the original problem can be greatly reduced. Next, we transform the complex fractional terms into rotated quadratic cones by introducing some slack variables. Finally, alternative and approximate optimization mechanisms with provable performance guarantees are explored to handle the mixed-integer and the non-convex properties of the sub-problems, respectively.

\section{Lyapunov-Based Optimization Framework}
Observe that there is a logarithmic term of the time-averaged achievable bitrates in the objective function, which significantly hinders the theoretical analysis of the formulated problem. To tackle this issue, we introduce a set of auxiliary variables to transform the logarithmic term.
In particular, define an auxiliary vector $\bm{\lambda} \left( t \right) = \left( {{\lambda _1}\left( t \right), \ldots ,{\lambda _M}\left( t \right)} \right)$, let $0 \le {\lambda _i}\left( t \right) \le r_i^{\max },\forall i,t$, and $\mathop {\lim \inf}\limits_{t \to \infty } \left[ {{{\bar r}_i}\left( t \right) - {{\bar \lambda }_i}\left( t \right)} \right] = 0,\forall i$, where ${\bar \lambda _i}\left( t \right) = \frac{1}{t}\sum\limits_{\tau  = 1}^t {{\lambda _i}\left( \tau  \right)} $, we then have $\mathop {\lim \inf}\limits_{t \to \infty }  \phi \left( {{{\bar \lambda }_i}\left( t \right)} \right) = \mathop {\lim \inf}\limits_{t \to \infty }  \phi \left( {{{\bar r}_i}\left( t \right)} \right)$. In addition, we define an auxiliary function $a\left( t \right) = \phi \left( \bm{\lambda} \left( t \right) \right) = \alpha \sum\limits_{i = 1}^M {{{\log }_2}\beta (1 + \frac{B{{\lambda _i}\left( t \right)}}{{{R_i}}} )}$.
By the Jensen’s inequality, the following expression holds, $\bar{a}\left( t \right) = \frac{1}{t}\sum\limits_{\tau  = 1}^t {a\left( \tau  \right)}  \le \phi \left( \bm{\bar \lambda} \left( t \right) \right)$. Based on this inequality, the original problem (\ref{eq:original_problem}) can be transformed into the following problem.
\begin{subequations}\label{eq:transformed_problem}
\begin{alignat}{2}
&\mathop {\rm Maximize}\limits_{{\mathcal{C}}\left( t \right),{\mathcal P}\left( t \right),{\mathcal Q}\left( t \right),\bm{\lambda} \left( t \right)} \mathop {\lim \inf}\limits_{t \to \infty } \bar a\left( t \right) - {\rho _1}\sum\limits_i {{d_i}\left( t \right)}  - {\rho _2}\sum\limits_{k \in {\mathcal K}} {\bar p_k^{tot}\left( t \right)}\\
&{\rm s.t.} \text{ } \mathop {\lim \inf}\limits_{t \to \infty } \left[ {{{\bar r}_i}\left( t \right) - {{\bar \lambda }_i}\left( t \right)} \right] = 0,\forall i\\
&\mathop {\lim \inf}\limits_{t \to \infty }  \left[ {{{\bar r}_i}\left( t \right) - {r_i^{th}}} \right] \ge 0,\forall i\\
&\mathop {\lim \sup}\limits_{t \to \infty }  \left[ {{{\tilde p}_k} - \bar p_k^{tot}\left( t \right)} \right] \ge 0,\forall k\\
&\rm{constraints\text{ }(\ref{eq:original_problem}d),\text{ }(\ref{eq:original_problem}e),\text{ }(\ref{eq:original_problem}f),\text{ }(\ref{eq:original_problem}g)\text{ }are\text{ }satisfied.}
\end{alignat}
\end{subequations}

\textbf{Remark 1:} On one hand, owing to the utilization of Jensen's inequality, the maximum value of the objective function of (\ref{eq:transformed_problem}) is not greater than that of (\ref{eq:original_problem}). On the other hand, (\ref{eq:transformed_problem}a) can obtain the maximum value of (\ref{eq:original_problem}a) by letting ${\bar \lambda _i}\left( t \right) = {\bar r_i}^ * \left( t \right),\forall i\in {\mathcal I},t$, with $\left\{ {{{\bar r}_1}^ * \left( t \right),{{\bar r}_2}^ * \left( t \right), \ldots {{\bar r}_M}^ * \left( t \right)} \right\}$ being a collection of optimal time-averaged achievable bitrates in (\ref{eq:original_problem}). Therefore, (\ref{eq:transformed_problem}) and (\ref{eq:original_problem}) are equivalent.

The transformed problem (\ref{eq:transformed_problem}), however, is still difficult to be solved due to the existence of time-averaged terms. To this end, the Lyapunov drift-plus-penalty technique is explored to handle the time-averaged terms.
We introduce three sets of virtual queues $\left\{ {{X_i}\left( t \right)} \right\}$, $\left\{ {{Z_i}\left( t \right)} \right\}$, $\left\{ {{Y_k}\left( t \right)} \right\}$, and define
\begin{equation}\label{eq:virtual_queue_X}
{X_i}\left( t \right) = {X_i}\left( {t - 1} \right) + {r_i^{th}} - {r_i}\left( {t - 1} \right),\forall i,t
\end{equation}
\begin{equation}\label{eq:virtual_queue_Z}
{Z_i}\left( t \right) = {Z_i}\left( {t - 1} \right) + {\lambda _i}\left( {t - 1} \right) - {r_i}\left( {t - 1} \right),\forall i,t
\end{equation}
\begin{equation}\label{eq:virtual_queue_Y}
{Y_k}\left( t \right) = {Y_k}\left( {t - 1} \right) + p_k^{tot}\left( {t - 1} \right) - {{\tilde p}_k},\forall k,t
\end{equation}

In this way, to enforce the three groups of time-averaged constraints
(\ref{eq:transformed_problem}b), (\ref{eq:transformed_problem}c), and (\ref{eq:transformed_problem}d), the following stability requirement should be satisfied.
 \begin{equation}\label{eq:stability_requirement_Y}
 \mathop {\lim }\limits_{t \to \infty } {{E\left\{ \max \left\{ {f\left( t \right),0} \right\} \right\}} \mathord{\left/
 {\vphantom {{E\left\{ {{{[{Y_k}\left( t \right)]}^ + }} \right\}} t}} \right.
 \kern-\nulldelimiterspace} t} = 0
\end{equation}
where $f(t) \in \{X_i(t),Z_i(t),Y_k(t); \forall i,k\}$.

According to the Lyapunov optimization approach, the Lyapunov function $L\left( t \right)$, which is nonnegative and can be regarded as a scalar measure of constraint violation at time slot $t$, is usually defined as the sum of square of all virtual queues. 
For the convenience of calculation, $L\left( t \right)$ is defined as
$L\left( t \right) \buildrel \Delta \over = \frac{1}{2}\sum\nolimits_i {{{( {{{[{X_i}( t )]}^ + }} )}^2}}  + \frac{1}{2}\sum\nolimits_i {{{( {{{[{Z_i}( t )]}^ + }} )}^2}} 
 + \frac{1}{2}\sum\nolimits_{k \in {\mathcal K}} {{{( {{{[{Y_k}( t )]}^ + }} )}^2}}$. 

Correspondingly, the expression of Lyapunov drift-plus-penalty can be written as  
$\Delta \left( t \right) - V(a\left( t \right) - {\rho _1}\sum\limits_i {{d_i}\left( t \right)}  - {\rho _2}\sum\limits_{k \in {\mathcal K}} {p_k^{tot}\left( t \right)} )$,
where $\Delta \left( t \right) = L\left( {t + 1} \right) - L\left( t \right)$ is a Lyapunov drift,  
$- (a\left( t \right) - {\rho _1}\sum\limits_i {{d_i}\left( t \right)}  - {\rho _2}\sum\limits_{k \in {\mathcal K}} {p_k^{tot}\left( t \right)} )$
is a penalty function, and $V$ is a non-negative coefficient, weighing a trade-off between the constraint violation and optimality. Thus, $V$ can be chosen properly to ensure that the time average of the penalty function is arbitrarily close to the optimum. 
According to (\ref{eq:virtual_queue_X}), (\ref{eq:virtual_queue_Z}), and (\ref{eq:virtual_queue_Y}), we can obtain the following expressions
\begin{equation}\label{eq:virtual_queue_X_square}
\begin{array}{l}
\frac{1}{2}{( {{{[{X_i}( {t + 1} )]}^ + }} )^2} = \frac{1}{2}{( {{{[{X_i}( t )]}^ + }} )^2} +\\
 {[{X_i}( t )]^ + }( {{r_i^{th}} - {r_i}( t )} ) + \frac{1}{2}{( {{r_i^{th}} - {r_i}( t )} )^2}
\end{array}
\end{equation}
\begin{equation}\label{eq:virtual_queue_Z_square}
\begin{array}{l}
\frac{1}{2}{( {{{[{Z_i}\left( {t + 1} \right)]}^ + }} )^2} = \frac{1}{2}{( {{{[{Z_i}\left( t \right)]}^ + }} )^2} +\\
 {[{Z_i}\left( t \right)]^ + }\left( {{\lambda _i}\left( t \right) - {r_i}\left( t \right)} \right) + \frac{1}{2}{\left( {{\lambda _i}\left( t \right) - {r_i}\left( t \right)} \right)^2}
\end{array}
\end{equation}
\begin{equation}\label{eq:virtual_queue_Y_square}
\begin{array}{l}
\frac{1}{2}{( {{{[{Y_k}\left( {t + 1} \right)]}^ + }} )^2} = \frac{1}{2}{( {{{[{Y_k}\left( t \right)]}^ + }} )^2} +\\
 {[{Y_k}\left( t \right)]^ + }\left( {{p_k}\left( t \right) + p_k^c - {{\tilde p}_k}} \right) + \frac{1}{2}{\left( {{p_k}\left( t \right) + p_k^c - {{\tilde p}_k}} \right)^2}
\end{array}
\end{equation}

With (\ref{eq:virtual_queue_X_square})-(\ref{eq:virtual_queue_Y_square}) and $\Delta \left( t \right) = L\left( {t + 1} \right) - L\left( t \right)$, the upper bound of the drift-plus-penalty function at time slot $t$ can be given by
\begin{equation}\label{eq:upper_bound_drift_plus_penalty}
\begin{array}{l}
\Delta \left( t \right) - V(a\left( t \right) - {\rho _1}\sum\limits_i {{d_i}\left( t \right)}  - {\rho _2}\sum\limits_{k \in {\mathcal K}} {p_k^{tot}\left( t \right)} )\\
 \le \sum\limits_i {{{\left( {r_i^{\max }} \right)}^2}}  + \sum\limits_{k \in {\mathcal K}} {\frac{{{{\left( {{{\hat p}_k}} \right)}^2}}}{2}}  + \sum\limits_i {{{[{X_i}\left( t \right)]}^ + }r_i^{th}} - \\
 \sum\limits_{k \in {\mathcal K}} {{{[{Y_k}\left( t \right)]}^ + }\left( {{{\tilde p}_k} - p_k^c} \right)}  + V{\rho _2}\sum\limits_{k \in {\mathcal K}} {p_k^c}  - V\phi \left( {\bm {\lambda} \left( t \right)} \right)+ \\
 \sum\limits_i {{{[{Z_i}\left( t \right)]}^ + }{\lambda _i}\left( t \right)}  + \sum\limits_{k \in {\mathcal K}} {\left\{ {V{\rho _2} + {{[{Y_k}\left( t \right)]}^ + }} \right\}{p_k}\left( t \right)}-\\
 \sum\limits_i {\left\{ {{{[{X_i}\left( t \right)]}^ + } + {{[{Z_i}\left( t \right)]}^ + }} \right\}{r_i}\left( t \right)}  + V{\rho _1}\sum\limits_i {{d_i}\left( t \right)} 
\end{array}
\end{equation}

As a consequence, the minimization of the drift-plus-penalty item can be approximated by minimizing its upper bound, i.e., the right-hand-side (RHS) expression of (\ref{eq:upper_bound_drift_plus_penalty}). Then, the optimization problem can be greedily solved by minimizing the upper bound of the drift-plus-penalty function at each time slot. From (\ref{eq:upper_bound_drift_plus_penalty}), we can also observe that its RHS expression can be decomposed into three types of independent items, including constant items, auxiliary variable-related items, and resource optimization-related items (exactly, the UAVs' transmit power ${\mathcal P}(t)$, the UAVs' trajectories ${\mathcal Q}(t)$, and the serving UAV selection ${\mathcal C}(t)$).

In summary, the Lyapunov-based optimization framework of mitigating (\ref{eq:transformed_problem}) can be decomposed into the following repeated optimization sub-problems of a two-layered structure. 
\begin{itemize}
\item \textbf{Auxiliary-variable-layer optimization:}
Optimize (\ref{eq:Auxiliary_variable_optimization}) to obtain the optimal ${\lambda _i}\left( t \right)$ for $\forall i \in {\mathcal I}$.
\begin{subequations}\label{eq:Auxiliary_variable_optimization}
\begin{alignat}{2}
&\mathop {\rm Minimize}\limits_{\bm{\lambda} \left( t \right)}  - V\phi \left(\bm {\lambda} \left( t \right) \right) + \sum\limits_i {{{[{Z_i}\left( t \right)]}^ + }} {\lambda _i}\left( t \right)\\
&{\rm s.t.} \text{ } 0 \le {\lambda _i}\left( t \right) \le r_i^{\max },\forall i \in {\mathcal I},t
\end{alignat}
\end{subequations}

\item \textbf{Resource-layer optimization:}
The serving UAV selection, the UAVs' transmit power, and the UAVs' trajectories will be obtained by solving the following multi-objective optimization sub-problem.
\begin{subequations}\label{eq:association_power_trajectories_optimization}
\begin{alignat}{2}
&\mathop {\rm Minimize}\limits_{{\mathcal C}\left( t \right),{\mathcal P}\left( t \right),{\mathcal Q}\left( t \right)} \sum\limits_{k \in {\mathcal K}} {\left\{ {V{\rho _2} + {{[{Y_k}\left( t \right)]}^ + }} \right\}{p_k}\left( t \right)} -\notag\\
& \sum\limits_i {\left\{ {{{[{X_i}\left( t \right)]}^ + } + {{[{Z_i}\left( t \right)]}^ + }} \right\}{r_i}\left( t \right)} + V{\rho _1}\sum\limits_i {{d_i}\left( t \right)} \\
&{\rm s.t.} \text{ } \rm{constraint \text{ }(\ref{eq:transformed_problem}e)\text{ }is\text{ }satisfied.}
\end{alignat}
\end{subequations}
\end{itemize}
\section{Problem solution and Algorithm Design} 
\subsection{Solution to the Auxiliary-Variable-Layer Sub-problem}
As the auxiliary function $\phi \left( \bm{\lambda} \left( t \right) \right)$ is the total of all individual logarithmic functions, we can then divide this sub-problem into $M$ individually optimized sub-problems, each of which can be formulated as
\begin{subequations}\label{eq:Auxiliary_variable_sub_problem}
\begin{alignat}{2}
&\mathop {\rm Minimize}\limits_{{\lambda _i}\left( t \right)}  - V\alpha {\log _2}\beta (1 + \frac{B{{\lambda _i}\left( t \right)}}{{{R_i}}})  + {[{Z_i}\left( t \right)]^ + }{\lambda _i}\left( t \right)\\
&{\rm s.t.} \text{ } 0 \le {\lambda _i}\left( t \right) \le r_i^{\max },\forall i \in {\mathcal I},t
\end{alignat}
\end{subequations}

(\ref{eq:Auxiliary_variable_sub_problem}a) is a convex function of ${\lambda _i}\left( t \right)$, and its closed-form solution can be obtained by calculating the derivative. Let $\frac{{\partial f}}{{\partial {\lambda _i}\left( t \right)}} = {[{Z_i}\left( t \right)]^ + } - \frac{{V\alpha }}{{\beta \left( {\frac{{{R_i}}}{B} + {\lambda _i}\left( t \right)} \right)\ln 2}} = 0$, we can obtain
\begin{equation}\label{eq:Auxiliary_variable_solution}
{\lambda _i}\left( t \right) = \left\{ {\begin{array}{*{20}{c}}
{r_i^{\max },}&{{{[{Z_i}\left( t \right)]}^ + } = 0}\\
{\min \left\{ {{{\left[ {\frac{{V\alpha \ln^{-1} 2}}{{{{[{Z_i}\left( t \right)]}^ + }\beta }} - \frac{{{R_i}}}{B}} \right]}^ + },r_i^{\max }} \right\},}&{\rm otherwise}
\end{array}} \right.
\end{equation}

\subsection{Solution to the Resource-Layer Sub-problem}
It can be observed that (\ref{eq:association_power_trajectories_optimization}) includes some logarithmic-quadratic terms and non-convex constraints. It is also a multi-objective optimization problem involving both integer and continuous decision variables. 
As a result, (\ref{eq:association_power_trajectories_optimization}) is difficult to be solved directly. To this end, an iterative optimization scheme is adopted to solve (\ref{eq:association_power_trajectories_optimization}) in this paper.

\subsubsection{Solution to serving UAV selection sub-problem} 
Given UAVs' trajectories ${\mathcal Q}(t)$ and UAVs' transmit power ${\mathcal P}(t)$, the serving UAV selection ${\mathcal C}(t)$ at time slot $t$ can be optimized by solving the following sub-problem. 
\begin{subequations}\label{eq:association_transform_problem}
\begin{alignat}{2}
&\mathop {\rm Maximize}\limits_{{\rm{\mathcal C}}\left( t \right)} \sum\limits_i {\left\{ {{{[{X_i}\left( t \right)]}^ + } + {{[{Z_i}\left( t \right)]}^ + }} \right\}} \times \notag \allowdisplaybreaks[4] \\
&\sum\limits_{k \in {\mathcal K}} {{c_{ik}}\left( t \right){{\log }_2}\left( {1 + sin{r_{ik}}\left( t \right)} \right)} - V{\rho _1} \times \notag \allowdisplaybreaks[4] \\
& \sum\limits_i {(\sum\limits_{k \in {\mathcal K}} {\frac{{L{c_{ik}}\left( t \right)}}{{B{{\log }_2}\left( {1 + sin{r_{ik}}\left( t \right)} \right)}}}  + (1 - \sum\limits_{k \in {\mathcal K}} {{c_{ik}}\left( t \right)} )\delta t)} \allowdisplaybreaks[4] \\
&{\rm s.t.} \text{ }  \rm{ constraints \text{ } (\ref{eq:tan_theta}),\text{ }(\ref{eq:association_value}),\text{ }(\ref{eq:original_problem}f)\text{ } are \text{ } satisfied.}
\end{alignat}
\end{subequations}

It can be confirmed that (\ref{eq:association_transform_problem}) is an integer linear programming problem, and some optimization tools such as MOSEK can be employed
to alleviate this sub-problem effectively.

\subsubsection{Solution to the UAVs' transmit power control sub-problem} For any given serving UAV selection ${\mathcal C}(t)$, UAVs' trajectories ${\mathcal Q}(t)$, and UAVs' transmit power ${\mathcal P}(t-1)$ at the previous time slot $t - 1$, the UAVs' transmit power ${\mathcal P}(t)$ can be optimized by solving the following sub-problem.
\begin{subequations}\label{eq:transmission_power_problem}
\begin{alignat}{2}
 &\mathop {\rm Minimize}\limits_{{\mathcal P}\left( t \right)} \sum\limits_{k \in {\mathcal K}} {\left\{ {V{\rho _2} + {{[{Y_k}\left( t \right)]}^ + }} \right\}{p_k}\left( t \right)} + \notag\\
&  V{\rho _1}\sum\limits_i {\frac{L}{B{{r_i}\left( t \right)}}}  - \sum\limits_i {\left\{ {{{[{X_i}\left( t \right)]}^ + } + {{[{Z_i}\left( t \right)]}^ + }} \right\}{r_i}\left( t \right)}\\
&{\rm s.t.} \text{ } \rm{constraints \text{ } (\ref{eq:original_problem}d), \text{ }(\ref{eq:original_problem}e) \text{ } are \text{ } satisfied.}
\end{alignat}
\end{subequations}

To simplify the expression of the objective function, we introduce the slack variable ${{\eta _i}\left( t \right)}$, and let ${\eta _i}\left( t \right) \le {r_i}\left( t \right)$. Moreover, we can observe that (\ref{eq:transmission_power_problem}a) includes a latency-related item with ${{r_i}\left( t \right)}$ being the denominator of the fraction, which greatly hinders the theoretical analysis of (\ref{eq:transmission_power_problem}). To tackle this issue, we introduce the slack variable ${\xi _i}\left( t \right)$ and let $\frac{L}{{B{\xi _i}\left( t \right)}} \le {\eta _i}\left( t \right)$. Specifically, we can transform the latency-related sub-problem in (\ref{eq:transmission_power_problem}a) into the following optimization problem.
\begin{subequations}\label{eq:latency_related_transform_problem}
\begin{alignat}{2}
&\mathop {\rm Minimize}\limits_{{\mathcal P}\left( t \right),{\xi _i}\left( t \right),{\eta _i}\left( t \right)} \text{ } \sum\limits_i {{\xi _i}\left( t \right)} \\
&{\rm s.t.} \text{ } {\xi _i}\left( t \right){\eta _i}\left( t \right) \ge \frac{L}{B},\forall i,t\\
&{\eta _i}\left( t \right) \le {r_i}\left( t \right),\forall i,t\\
&{\eta _i}\left( t \right) \ge 0,{\xi _i}\left( t \right) \ge 0,\forall i,t
\end{alignat}
\end{subequations}

According to  (\ref{eq:achievable_rate_subscriber_i}) and (\ref{eq:latency_related_transform_problem}), the UAVs' transmit power optimization sub-problem (\ref{eq:transmission_power_problem}) can be reformulated as
\begin{subequations}\label{eq:transmission_power_transform_problem}
\begin{alignat}{2}
&\mathop {\rm Maximize}\limits_{{\mathcal P}\left( t \right),{\xi _i}\left( t \right),{\eta _i}\left( t \right)}  - \sum\limits_{k \in {\mathcal K}} {\left\{ {V{\rho _2} + {{[{Y_k}\left( t \right)]}^ + }} \right\}{p_k}\left( t \right)}  + \notag\\
&\sum\limits_i {\left\{ {{{[{X_i}\left( t \right)]}^ + } + {{[{Z_i}\left( t \right)]}^ + }} \right\}{\eta _i}\left( t \right)}  - V{\rho _1}\sum\limits_i {{\xi _i}\left( t \right)} \\
&{\rm s.t.} \text{ } \sum\limits_{k \in {\mathcal K}} {{c_{ik}}\left( t \right){{\log }_2} (1 + \frac{{{p_k}\left( t \right){h_{ik}}\left( t \right)}}{{{\sigma ^2}{\rm{ + }}\sum\limits_{j \in {\mathcal K}\backslash \left\{ k \right\}} {{p_j}\left( t \right){h_{ij}}\left( t \right)} }} )}\notag\\ 
&\ge {\eta _i}\left( t \right),\forall i,t\\
& \rm{\text{ }constraints \text{ }(\ref{eq:transmission_power_problem}b),\text{ } (\ref{eq:latency_related_transform_problem}b),\text{ }  (\ref{eq:latency_related_transform_problem}d)\text{ } are \text{ } satisfied.}
\end{alignat}
\end{subequations}

The constraint (\ref{eq:transmission_power_transform_problem}b), however, is non-convex, and thus (\ref{eq:transmission_power_transform_problem}) is a non-convex optimization sub-problem. To solve this problem effectively, we need to approximate the non-convex constraint as a convex constraint \cite{scutari2017parallel}. 
The following Proposition presents the approximation results.
\begin{propo}\label{propo_1}
\rm Given a local point $P_j^{(r)}(t)$, $\forall j,t$, one can approximately transform the non-convex constraint (\ref{eq:transmission_power_transform_problem}b) into the following convex one
\begin{equation}\label{eq:transmission_power_constraint_simplify}
\begin{array}{l}
\sum\limits_{k \in {\mathcal K}} {{c_{ik}}\left( t \right)( {{\hat \Lambda }_{i}}\left( t \right) - F_{ik}^{\left( r \right)}\left( t \right) )} - \\
  \sum\limits_{k \in {\mathcal K}} {( {c_{ik}}\left( t \right)\sum\limits_{j \in {\mathcal K}\backslash \left\{ k \right\}} {G_{ik}^{\left( r \right)}\left( t \right)( {p_j}\left( t \right) - p_j^{\left( r \right)}\left( t \right) )} )} \\
 \ge {\eta _i}\left( t \right),\forall i,t
\end{array}
\end{equation}
{\rm where, ${\hat \Lambda _{i}}\left( t \right) = {\log _2}( {\sigma ^2} + \sum\limits_{j \in {\mathcal K}} {{p_j}\left( t \right){h_{ij}}\left( t \right)} )$, $F_{ik}^{\left( r \right)}\left( t \right) = {\log _2}( {\sigma ^2}{\rm{ + }}\sum\limits_{j \in {\mathcal K}\backslash \left\{ k \right\}} {p_j^{\left( r \right)}\left( t \right){h_{ij}}\left( t \right)}  )$, and $G_{ik}^{\left( r \right)}\left( t \right) = \frac{{{h_{ij}}\left( t \right)}}{{( {\sigma ^2}{\rm{ + }}\sum\limits_{j \in {\mathcal K}\backslash \left\{ k \right\}} {p_j^{\left( r \right)}\left( t \right){h_{ij}}\left( t \right)} )\ln 2}} = \frac{{{h_{ij}}\left( t \right)}}{{{2^{F_{ik}^{\left( r \right)}\left( t \right)}}\ln 2}}$.}
\end{propo}
\begin{proof}
    Please refer to Appendix A.
\end{proof}

Based on {the relaxation transformation results and the approximation} results in Proposition \ref{propo_1}, 
{(\ref{eq:transmission_power_problem})} can be approximately transformed into a convex sub-problem. To facilitate the optimization of the approximated convex sub-problem using some convex optimization tools such as MOSEK, we further transform it into a conic problem in the following Lemma.

\begin{lemma}\label{lemma:lemma_transmission_power_optimization}
{\rm 
Based on the above approximation results, {(\ref{eq:transmission_power_problem})} can be approximately transformed into the following conic problem by introducing a family of auxiliary variables.
\begin{subequations}\label{eq:transmission_power_slack_problem_mosek}
\begin{alignat}{2}
&\mathop {\rm Maximize}\limits_{{\mathcal P}\left( t \right),{\xi _i}\left( t \right),{\eta _i}\left( t \right)}  - \sum\limits_{k \in {\mathcal K}} {\left\{ {V{\rho _2} + {{[{Y_k}\left( t \right)]}^ + }} \right\}{p_k}\left( t \right)}  + \notag\\
&\sum\limits_i {\left\{ {{{[{X_i}\left( t \right)]}^ + } + {{[{Z_i}\left( t \right)]}^ + }} \right\}{\eta _i}\left( t \right)}  - V{\rho _1}\sum\limits_i {{\xi _i}\left( t \right)} \\
& \rm{subject\text{ } to}: \notag \\
& \rm{linear\text{ }constraints}: 
(\ref{eq:transmission_power_problem}b),\text{ }
(\ref{eq:latency_related_transform_problem}d),\text{ }
(\ref{eq:transmission_power_mosek}a),\text{ }
(\ref{eq:transmission_power_mosek}b),\text{ } 
(\ref{eq:transmission_power_mosek}c) \\
& \rm{exponential\text{ } cone\text{ }constraint}:(\ref{eq:transmission_power_mosek_b_exponential_cone})\\
& \rm{rotated\text{ } quadratic\text{ } cone\text{ }constraint}: (\ref{eq:transmission_power_mosek_d_rotated_cone})
\end{alignat}
\end{subequations}

{Further, the opposite value of the maximum value of (\ref{eq:transmission_power_slack_problem_mosek}a) is the upper bound of the optimal objective value of (\ref{eq:transmission_power_problem}).}
}
\end{lemma}
\begin{proof}\renewcommand{\qedsymbol}{}
{Please refer to Appendix B.}
\end{proof}

\subsubsection{Solution to the UAV trajectory optimization sub-problem} Given serving UAV selection ${\mathcal C}(t)$, UAVs’ transmit power ${\mathcal P}(t)$, and UAVs’ trajectories ${\mathcal Q}(t-1)$ at the previous time slot $t - 1$, the UAVs’ trajectories ${\mathcal Q}(t)$ can be optimized by mitigating the following problem
\begin{subequations}\label{eq:trajectories_problem}
\begin{alignat}{2}
&\mathop {\rm Maximize}\limits_{{\mathcal Q}\left( t \right)} \sum\limits_i {\left\{ {{{[{X_i}\left( t \right)]}^ + } + {{[{Z_i}\left( t \right)]}^ + }} \right\}{r_i}\left( t \right)} - \notag \\
& V{\rho _1}\sum\limits_i {\frac{L}{B{{r_i}\left( t \right)}}} \\
&{\rm s.t.} \text{ }  \rm{\text{ }constraints \text{ } (\ref{eq:tan_theta}),\text{ }(\ref{eq:minimum_safety_distance}),\text{ }(\ref{eq:maximum_flight_speed})\text{ } are \text{ } satisfied.}
\end{alignat}
\end{subequations}

{Similar to the derivation in subsection 4.2.2,} after introducing the slack variables ${{\eta }_{i}}\left( t \right)$ and ${{\xi }_{i}}\left( t \right)$, the optimization problem (\ref{eq:trajectories_problem}) can be reformulated as
\begin{subequations}\label{eq:trajectories_slack_problem}
\begin{alignat}{2}
&\mathop {\rm Maximize}\limits_{{\mathcal Q}\left( t \right),{\eta _i}\left( t \right),{\xi _i}\left( t \right)} \sum\limits_i {\left\{ {{{[{X_i}\left( t \right)]}^ + } + {{[{Z_i}\left( t \right)]}^ + }} \right\}{\eta _i}\left( t \right)} - \notag\\
 & V{\rho _1}\sum\limits_i {{\xi _i}\left( t \right)} \\
&{\rm s.t.} \text{ } \sum\limits_{k \in {\mathcal K}} {{c_{ik}}\left( t \right){{\log }_2}(1 + \frac{{{p_k}\left( t \right){h_{ik}}\left( t \right)}}{{{\sigma ^2}{\rm{ + }}\sum\limits_{j \in {\mathcal K}\backslash \left\{ k \right\}} {{p_j}\left( t \right){h_{ij}}\left( t \right)} }} )} \notag\\
&\ge {\eta _i}\left( t \right),\forall i,t\\
&\rm{constraints \text{ }(\ref{eq:latency_related_transform_problem}b),\text{ }(\ref{eq:latency_related_transform_problem}d),\text{ } (\ref{eq:trajectories_problem}b) \text{ }are \text{ }satisfied.}
\end{alignat}
\end{subequations}

(\ref{eq:trajectories_slack_problem}) is not convex due to the existence of non-convex constraints (\ref{eq:minimum_safety_distance}) and (\ref{eq:trajectories_slack_problem}b). Therefore, it is quite difficult to obtain its optimal solution. 

For constraint (\ref{eq:minimum_safety_distance}), it is non-convex. Similar to the analysis and derivation in Proposition \ref{propo_1}, the first-order Taylor expansion can be performed to calculate the lower bound of its left-hand-side (LHS) term.
\begin{equation}\label{eq:minimum_safety_distance_sca}
\begin{array}{l}
{\left\| {\bm{q_k}\left( t \right) - \bm{q_j} \left( t \right)} \right\|^2} \ge  - {\left\| {\bm{q_k}^{\left( r \right)}\left( t \right) - \bm{q_j}^{\left( r \right)}\left( t \right)} \right\|^2}\\
 + 2{\left( {\bm{q_k}^{\left( r \right)}\left( t \right) - \bm{q_j}^{\left( r \right)}\left( t \right)} \right)^{\rm{T}}}\left( {\bm{q_k}\left( t \right) - \bm{q_j} \left( t \right)} \right)
\end{array}
\end{equation}
where $\bm{q_k}^{\left( r \right)}\left( t \right)$ and $\bm{q_j}^{\left( r \right)}\left( t \right)$ denote the 2D horizontal location of the $k$-th and $j$-th UAV at the $r$-iteration of the approximation method explored in Proposition \ref{propo_1}, respectively.

For the non-convex constraint (\ref{eq:trajectories_slack_problem}b), it is much more complex than (\ref{eq:minimum_safety_distance}). The following Proposition presents its approximated constraints.
\begin{propo}\label{propo_2}
{\rm {Given a local point ${\bm q}_j^{(r)}(t)$, $\forall j,t$, by introducing a slack variable $B_{ij}(t) \le ||{{\bm q_j}(t)-{\bm s_i}(t)}||^2$, we can approximately transform the non-convex constraint (\ref{eq:trajectories_slack_problem}b) into the following convex constraints}
\begin{equation}\label{eq:trajectories_problem_constraint_final}
\begin{array}{l}
\sum\limits_{k \in{\mathcal K}} {{c_{ik}}\left( t \right)} (D_i^{\left( r \right)}\left( t \right) - \sum\limits_{j \in {\mathcal K}} {E_{ij}^{\left( r \right)}\left( t \right)({{\left\| {\bm{q_j}\left( t \right) - \bm{s_i}\left( t \right)} \right\|}^2} - } \\
{\left\| {\bm{q_j}^{\left( r \right)}\left( t \right) - \bm{s_i}\left( t \right)} \right\|^2})) + \sum\limits_{k \in {\mathcal K}} {{c_{ik}}\left( t \right){{\mathord{\tilde 
\Lambda } }_{ik}}\left( t \right)}  \ge {\eta _i}\left( t \right),\forall i,t
\end{array}
\end{equation}
\begin{equation}\label{eq:slack_variable_sca}
\begin{array}{l}
{\left\| {\bm{q_j}^{\left( r \right)}\left( t \right) - \bm{s_i}\left( t \right)} \right\|^2}
 + 2{( \bm{q_j}^{\left( r \right)}\left( t \right) - \bm{s_i}\left( t \right) )^{\rm{T}}} \times \\ \left( {\bm{q_j}\left( t \right) - \bm{s_i}\left( t \right)} \right) \ge B_{ij}(t)
\end{array}
\end{equation}
where $D_i^{\left( r \right)}\left( t \right){\rm{ = }}{\log _2}( {\sigma ^2} + \sum\limits_{j \in {\mathcal K}} {\frac{{{p_j}\left( t \right){\omega _{ij}}}}{{{H^2} + {{\left\| {\bm{q_j}^{\left( r \right)}\left( t \right) - \bm{s_i}\left( t \right)} \right\|}^2}}}} )$, $E_{ij}^{\left( r \right)}\left( t \right){\rm{ = }}\frac{{{p_j}\left( t \right){\omega _{ij}}}}{{{{\left( {{H^2} + {{\left\| {\bm{q_j}^{\left( r \right)}\left( t \right) - \bm{s_i}\left( t \right)} \right\|}^2}} \right)}^2}D_i^{\left( r \right)}\left( t \right)\ln 2}}$, and ${\mathord{\tilde \Lambda } _{ik}}\left( t \right) =  - {\log _2}( {\sigma ^2} + \sum\limits_{j \in {\mathcal K}\backslash \left\{ k \right\}} {\frac{{{p_j}\left( t \right){\omega _{ij}}}}{{{H^2} + B_{ij}(t)}}}  )$.} 
\end{propo}
\begin{proof}\renewcommand{\qedsymbol}{}
{Please refer to Appendix C.}
\end{proof}

Based on the above derivation, we can approximately transform the non-convex UAV trajectory optimization sub-problem (\ref{eq:trajectories_problem}) into a convex one. We further transform the approximate sub-problem into a conic problem in the following Lemma, which can be effectively solved by calling MOSEK.
\begin{lemma}\label{lemma:lemma_trajectories_optimization}
{\rm Based on the above approximation results, (\ref{eq:trajectories_problem}) can be approximately transformed into the following conic problem by introducing a family of slack variables.
\begin{subequations}\label{eq:trajectories_slack_problem_mosek}
\begin{alignat}{2}
&\mathop {\rm Maximize}\limits_{{\mathcal Q}\left( t \right),{\eta _i}\left( t \right),{\xi _i}\left( t \right),B_{ij}(t)} \sum\limits_i {\left\{ {{{[{X_i}\left( t \right)]}^ + } + {{[{Z_i}\left( t \right)]}^ + }} \right\}{\eta _i}\left( t \right)} - \notag\\
 & V{\rho _1}\sum\limits_i {{\xi _i}\left( t \right)} \\
& \rm{subject\text{ } to}: \notag \\
& \rm{linear\text{ }constraints}: (\ref{eq:slack_variable_sca}), \text{ }(\ref{eq:lemma_trajectory_slack_2_transform_cone}a),\text{ } 
(\ref{eq:lemma_trajectory_slack_3}),\text{ }
(\ref{eq:trajectories_problem_constraint_final_linear}),\text{ } 
(\ref{eq:minimum_safety_distance_sca_linear})\\ 
& \rm{quadratic \text{ } cone\text{ }constraints}:
(\ref{eq:lemma_tan_theta_cone}),\text{ }
(\ref{eq:lemma_maximum_flight_speed_cone})\\
& \rm{rotated\text{ } quadratic\text{ } cone\text{ }constraints}: 
(\ref{eq:transmission_power_mosek_d_rotated_cone}),\text{ }
(\ref{eq:lemma_trajectory_slack_1_transform_cone})\\
& \rm{exponential\text{ } cone\text{ }constraints}:(\ref{eq:lemma_trajectory_slack_2_transform_cone}b),\text{ } (\ref{eq:lemma_trajectory_slack_3_cone})
\end{alignat}
\end{subequations}

{Further, the maximum value of (\ref{eq:trajectories_slack_problem_mosek}a) is the lower bound of the optimal objective value of (\ref{eq:trajectories_problem}).}
}
\end{lemma}
\begin{proof}\renewcommand{\qedsymbol}{}
{Please refer to Appendix D.}
\end{proof}
\subsection{Algorithm Design}
Based on the above theoretical analysis and derivation, we can summarize the main steps of solving (\ref{eq:association_power_trajectories_optimization}) in the following algorithm. Besides, the following Lemma presents the convergence of the algorithm. 
\begin{algorithm}
\caption{multi-UAV network resource-layer iterative optimization algorithm}
\label{alg:alg1}
\begin{algorithmic}[1]
\STATE \textbf{Initialization:} Randomly initialize ${\mathcal P}^{(0)}(t)$ and ${\mathcal Q}^{(0)}(t)$, let $r = 0$.
\REPEAT
\STATE Given ${\mathcal P}^{(r)}(t)$ and ${\mathcal Q}^{(r)}(t)$, solve (\ref{eq:association_transform_problem}) to obtain the solution ${{\mathcal C}^{(r+1)}(t)}$.
\STATE Given ${\mathcal C}^{(r+1)}(t)$, ${\mathcal P}^{(r)}(t)$, and ${\mathcal Q}^{(r)}(t)$, solve (\ref{eq:transmission_power_slack_problem_mosek}) to achieve the solution ${{\mathcal P}^{(r+1)}(t)}$.
\STATE Given ${\mathcal C}^{(r+1)}(t)$, ${\mathcal P}^{(r+1)}(t)$, and ${\mathcal Q}^{(r)}(t)$, solve (\ref{eq:trajectories_slack_problem_mosek}) to obtain the solution ${{\mathcal Q}^{(r+1)}(t)}$.
\STATE Update $r = r + 1$
\UNTIL {Convergence or reach the maximum number of iteration $r_{max}$.}
\end{algorithmic}
\end{algorithm}

{
\begin{lemma}\label{lemma:lemma_algorithm_convergence}
{\rm The iterative optimization Algorithm \ref{alg:alg1} is convergent.}
\end{lemma}
\begin{proof}\renewcommand{\qedsymbol}{}
{Please refer to Appendix E.}
\end{proof}
}
Recall that (\ref{eq:original_problem}) can be decomposed into two repeated optimization sub-problems. We obtained the closed-form solution of the auxiliary-variable-layer sub-problem. The resource-layer sub-problem can be solved by calling Algorithm \ref{alg:alg1}. We can then summarize the steps of solving (\ref{eq:original_problem}) in Algorithm \ref{alg:alg2}.
\begin{algorithm}
\caption{Energy-efficient Multi-UAV network Optimization, EMUO}
\label{alg:alg2}
\begin{algorithmic}[1]
\STATE \textbf{Initialization:} Randomly initialize $X_i(1)$, $Z_i(1)$, and $Y_k(1)$ to positive values for any subscriber $i \in {\mathcal I}$ and UAV $k \in {\mathcal K}$. 
\FOR {each time slot $t = 1, 2, \ldots, T$}
\STATE Observe the virtual queues $X_i(t)$, $Z_i(t)$, and $Y_k(t)$.
\STATE Compute ${\lambda _i}\left( t \right)$ using (\ref{eq:Auxiliary_variable_solution}) for any subscriber $i$.
\STATE Find the serving UAV selection set ${\mathcal C}(t)$, UAV trajectories ${\mathcal Q}(t)$, and UAV transmit power ${\mathcal P}(t)$ by calling Algorithm \ref{alg:alg1}.
\STATE Calculate ${r_i}\left( t \right)$ for any subscriber $i \in {\mathcal I}$ using (\ref{eq:achievable_rate_subscriber_i}).
\STATE Calculate ${d_i}\left( t \right)$ for any subscriber $i \in {\mathcal I}$  using (\ref{eq:latency_subscriber_i}).
\STATE Calculate $p_j^{tot}(t)$ for each UAV $k \in {\mathcal K}$ using (\ref{eq:total_power_at_t}).
\STATE Update $X_i(t+1)$, $Z_i(t+1)$, and $Y_k(t+1)$ for any subscriber $i \in {\mathcal I}$ and UAV $k \in {\mathcal K}$ using (\ref{eq:virtual_queue_X}), (\ref{eq:virtual_queue_Z}), and (\ref{eq:virtual_queue_Y}), respectively.
\ENDFOR
\end{algorithmic}
\end{algorithm}


The computational complexity of Algorithm \ref{alg:alg2} has two main contributors at each time slot, i.e., the auxiliary-variable-layer optimization and the resource-layer optimization. The computational complexity of solving the auxiliary-variable-layer optimization sub-problem is $O(N)$. For the resource-layer optimization, it can be further decomposed into three sub-problems, including serving UAV selection optimization, UAV transmit power optimization, and UAV trajectory optimization sub-problems. The computational complexity of solving the linear integer serving UAV selection optimization sub-problem by a branch-and-bound method is $\min \{O((1+M)^N), O((1+N)^M)\}$. Both UAV transmit power optimization and trajectory optimization sub-problems are approximately transformed into convex programming; thus, the computational complexities of solving the approximate UAV transmit power optimization and trajectory optimization sub-problems by an interior method are $O((N+2M)^{3.5})$ and $O((N+2M+NM)^{3.5})$, respectively \cite{ye2011interior}. Besides, an iterative optimization scheme is explored to solve the resource-layer optimization sub-problem, and hence, the computational complexity of solving the resource-layer optimization sub-problem is $O(r_{\max}(\min \{O((1+M)^N), O((1+N)^M)\}+ O((N+2M)^{3.5})+ O((N+2M+NM)^{3.5})))$ in the worst-case. 

\section{Simulation and Result Analysis}
\subsection{Comparison Algorithms and Parameter Setting}
In this section, we conduct simulations to verify the effectiveness of the proposed algorithm, and compare it with the following five benchmark algorithms:
    1) Nearest Neighbor ASsociation (NNAS) algorithm: It implements the UAV trajectory and transmit power optimization method in \cite{Zhan2021Joint}. Meanwhile, each UAV only establishes a communication link (if available) with its nearest ground subscriber in the NNAS algorithm. 
    2) Stationary UAV DEployment (SUDE) algorithm: UAVs hover steadily over the randomly generated locations, as in \cite{Chen2020Optimal}. Meanwhile, the serving UAV selection and UAV transmit power are optimized using the schemes designed in the proposed algorithm. 
    3) Stationary UAV with maximum transmit power (SUMTP) algorithm: In SUMTP, UAVs hover steadily over a random selection of locations and deliver video streams with the maximum transmit power. Besides, it adopts the same serving UAV selection scheme as the proposed algorithm. 
    4) Circular UAV TRajectory (CUTR) algorithm: {As in \cite{Yang2022Fresh}}, each UAV flies in a circular trajectory with a speed of two m/s. The distance between any two neighboring UAVs is $1/4N$ ($N$ is the number of deployed UAVs) km at the initial time. Besides, CUTR adopts the similar serving UAV selection and UAV transmit power optimization schemes as the proposed algorithm. 
    5) Circular UAV trajectory with the maximum transmit power (CUMTP) algorithm: The difference between CUMTP and CUT is that CUMTP adopts the scheme of the maximum transmit power.

We consider a mountainous suburb area $\mathbb R$ of size $500 \times 500$ m$^2$, where ground subscribers walk randomly in the area. 
The radio frequency propagation parameters are: carrier frequency {{$f_c = 4.9$} GHz}, light speed {$c = 3.0\times 10^8$ m/s}, {noise power {$\sigma^2 = -174$ dBm}}, total bandwidth {{$B = 100$} MHz},
far field reference distance {$D_0 = 1$ m}, antenna transmitting gains {$g_{ik}^{T} = 1$}, and antenna receiving gains {$g_{ik}^{R} = 1$}.
The values of parameters related to UAVs are set: the maximum instantaneous total power {${{\hat p}_k} = 500$ mW}, the maximum time-averaged total power {${{\tilde p}_k} = 450$ mW}, circuit power {$p_k^c = 20$ mW}, the minimum safety distance {$d_{min} = 50$ m}, the maximum flight distance in one time slot {$s_{max} = 250$ m}, elevation angle threshold $\theta  = {77^{ \circ }}$, and fixed flight altitude {$H = 500$ m}. 
Several parameters related to video transmission are: the length of transmitted video data in a time slot {$L = 10$ Mb}, a turntable game in \cite{Yang2022Fresh} is called to set the required playback bitrate of subscriber $i$ with $R_i \in \{0.0316, 0.0154\}$ bps/Hz, let the minimum time-averaged achievable bitrate ${r_i^{th}} =R_i$, and the duration of a time slot ${\delta t}=2$ s.
More system parameters are listed as below: the video streaming utility model parameters {$\alpha = 1$} and {$\beta = 1$}; {$V = 10$}, {${\rho _1} = 15$}, {${\rho _2} = 0.05$},  {$T = 200$}, and the maximum number of optimization iteration {$r_{max} = 60$}. 

\subsection{Performance Evaluation}
In this subsection, we design extensive simulations to comprehensively verify the performance of the proposed algorithm, including the stability verification, the QoE performance, and the energy efficiency of the algorithm. 
To weaken the impact of randomly initialized parameters (e.g., UAV transmit power and locations) on the performance of algorithms, we run all algorithms for ten times, and then their average values are compared.
\subsubsection{{Results of stability and UAVs' trajectories}}
In this simulation, we test the stability of proposed algorithm. The stability refers to the stability of the introduced virtual queues and is defined as ${S_X} = {\max _i}{[{X_i}\left( t \right)]^ + }/t$, ${S_Z} = {\max _i}{[{Z_i}\left( t \right)]^ + }/t$, and ${S_Y} = {\max _k}{[{Y_k}\left( t \right)]^ + }/t$.
\begin{figure}[!t]
\centering
\includegraphics[width=3.2 in, height = 2.0 in]{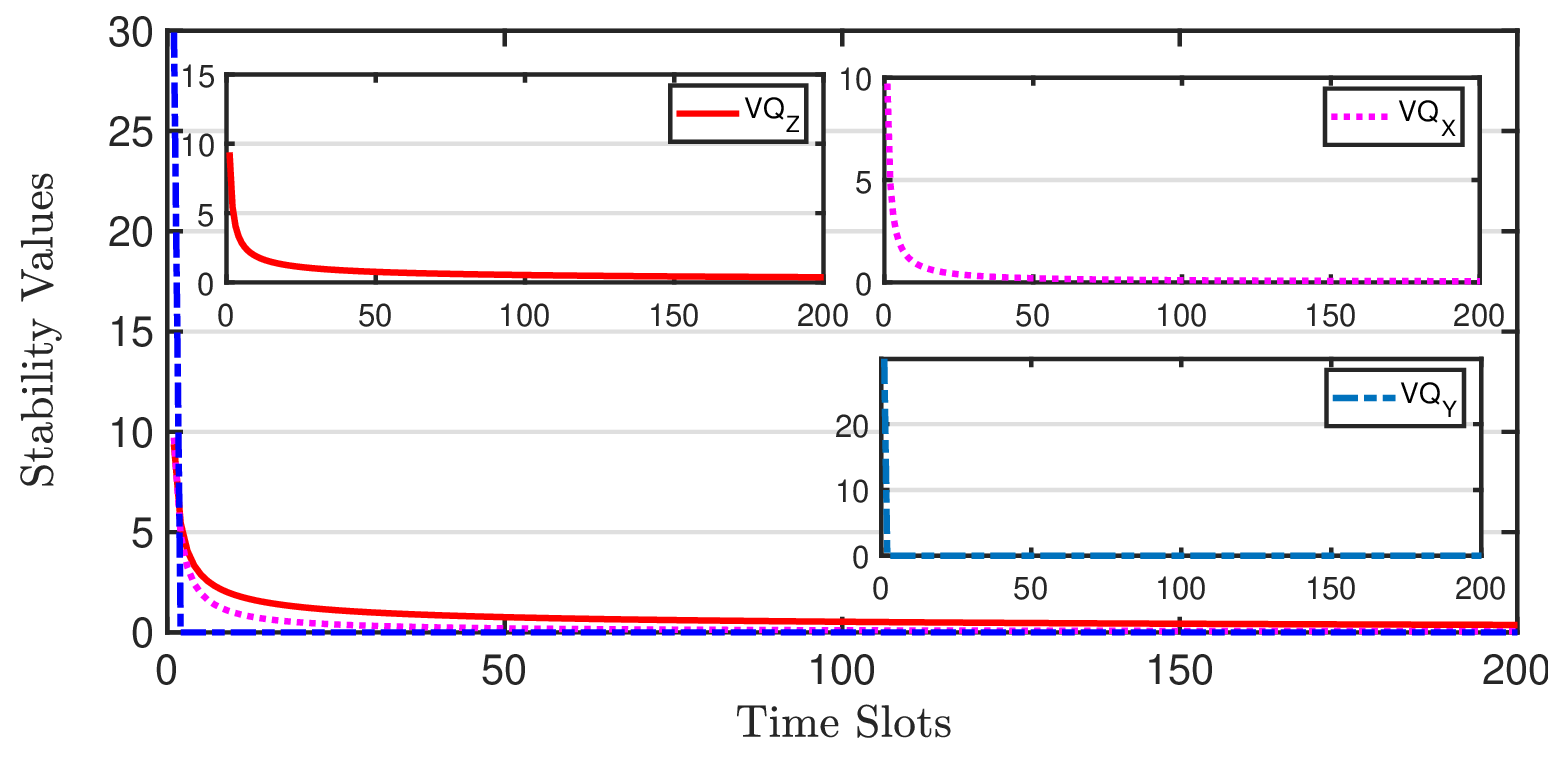}
\caption{Stability trends of virtual queues vs. time slots.}
\label{fig_virtual_queue_stability}
\end{figure}
In Fig. \ref{fig_virtual_queue_stability}, we plot the stability trends of the introduced virtual queues.
From this figure it can be observed that the obtained stability values of virtual queues are upper bounded over the whole period and tend to zero as time slot $t$ increases. 
According to the definition of mean-rate stability, we can say that the virtual queues are mean-rate stable, and then the time average-related constraint (\ref{eq:stability_requirement_Y}) can be satisfied. It also indicates that the frame freezing can be alleviated.
\begin{figure}[!t]
\centering
\includegraphics[width=3.2 in, height = 2.0 in]{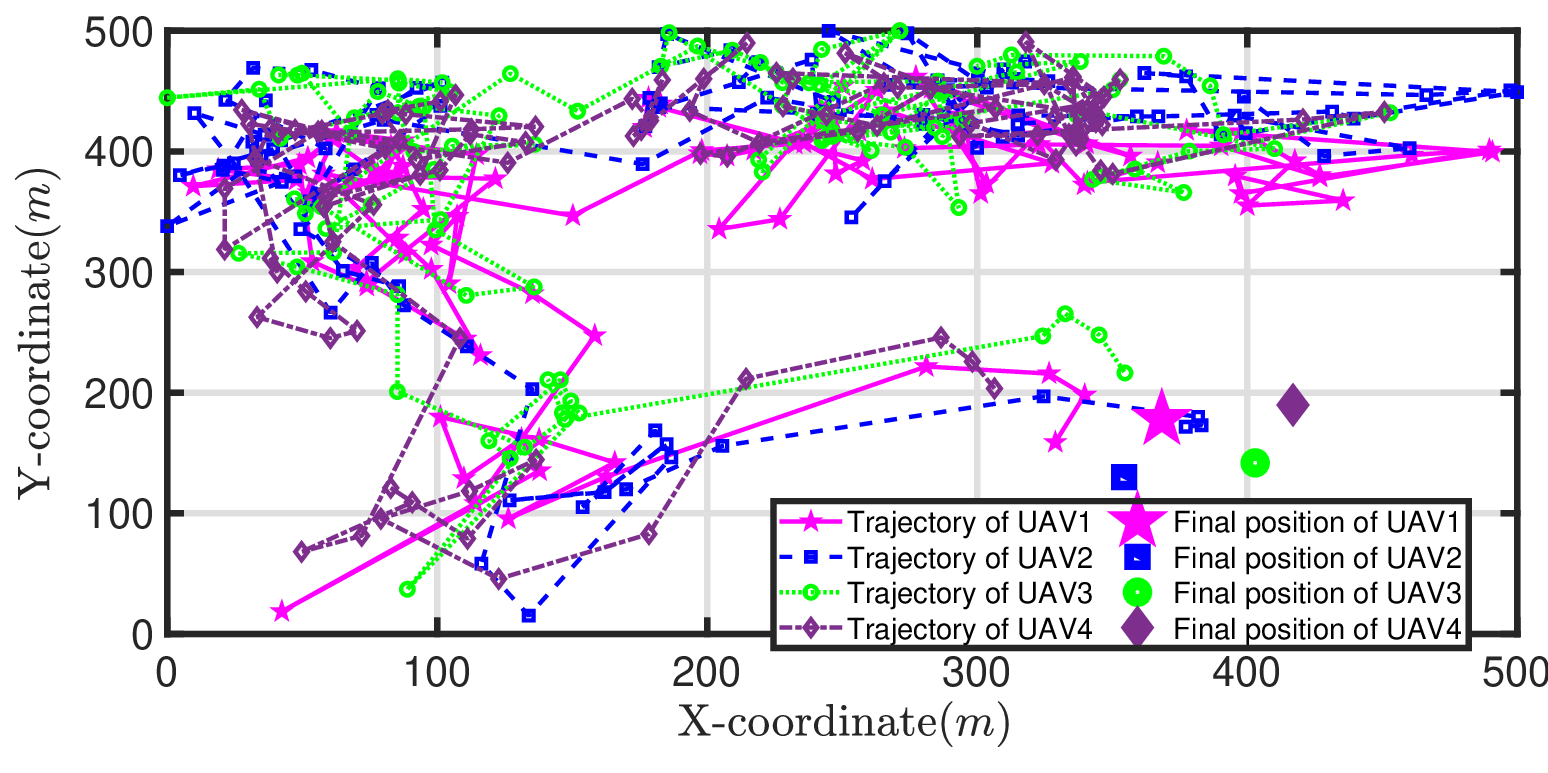}
\caption{{An illustration of trajectories and final positions of four UAVs.}}
\label{fig_trajectory_2D}
\end{figure}

{Further, Fig. \ref{fig_trajectory_2D} shows the trajectories of four UAVs in the first 100 time slots and their final positions.}

\subsubsection{Results of QoE}
The QoE performance of the proposed algorithm is verified by comparing it with other five benchmarks. 
As the frame freezing constraint will not be violated, we mathematically define the QoE as $QoE = NP - {\rho _1} * TL$ in this simulation, where the network profit $NP = \phi \left( {\bar r\left( T \right)} \right) = \alpha \sum\limits_{i = 1}^M {{{\log }_2}\beta (1 + \frac{{B{{\bar r}_i}\left( T \right)}}{{{R_i}}} )}$ and $TL$ represents the total latency. Recall the definition of the objective function, the latency experienced by an unassociated subscriber in a time slot is $\delta t$, the total latency can then be calculated by 
{$TL = \sum\limits_{i = 1}^M {\left\{ {\frac{1}{T}\sum\limits_{t = 1}^T {\left[ {\sum\limits_{k \in {\mathcal K}} {\frac{{L{c_{ik}}\left( t \right)}}{{B{{\log }_2}\left( {1 + sin{r_{ik}}\left( t \right)} \right)}}}  + (1 - \sum\limits_{k \in {\mathcal K}} {{c_{ik}}(t)} )\delta t} \right]} } \right\}}$}. 

\textbf{Network Profit:}
In Fig. \ref{fig_network_profit}, we plot the tendency of the achieved network profit versus the number of UAVs and the number of subscribers.
From this figure we have the following observations: 
1) The proposed EMUO algorithm achieves a high network profit, and it is hard to conclude the tendency of the achieved network profit when varying the number of UAVs. On one hand, more UAVs indicate that more subscribers can simultaneously receive video streams, and then greater network profit might be achieved. On the other hand, the signal interference becomes stronger and the corresponding achievable bitrate of each subscriber is reduced when increasing the number of UAVs; thus, a smaller network profit might also be obtained. Mainly owing to the joint trajectory optimization and power optimization, EMUO achieves greater network profit with the increase of the number of subscribers. 
2) For NNAS, it is unable to provide services for subscribers fairly and achieves the smallest network profit. Further, its achieved network profit decreases with an increasing number of UAVs due to stronger signal interference. Yet, NNAS has the potential to achieve higher profit when there are more subscribers in the considered area. 
3) For CUTR, its obtained network profit will increase with $N$ when $N \le 4$. Nevertheless, when $N > 4$, strong signal interference leads to the decrease in the achievable bitrates of subscribers and thus reduces the profit. 
Diverse from CUTR, CUMTP utilizes the maximum transmit power to deliver video streams. However, constrained by the LoS propagation condition, the coverage range of a UAV cannot be extended by choosing the maximum transmit power. Meanwhile, CUMTP will cause strong signal interference, and thus low network profit is obtained. 
4) Similar to the two circular UAV trajectory-based algorithms, SUDE achieves higher network profit than SUMTP. 
The comparison results indicate that maximizing the transmit power of UAVs cannot effectively increase the network profit, and it is essential to optimize the UAV transmit power. 
\begin{figure}[!tb]
\centering
  \subfigure[Network profit vs. number of UAVs]{
    \label{fig:subfig:a} 
    \includegraphics[width=3.2 in, height = 2.0 in]{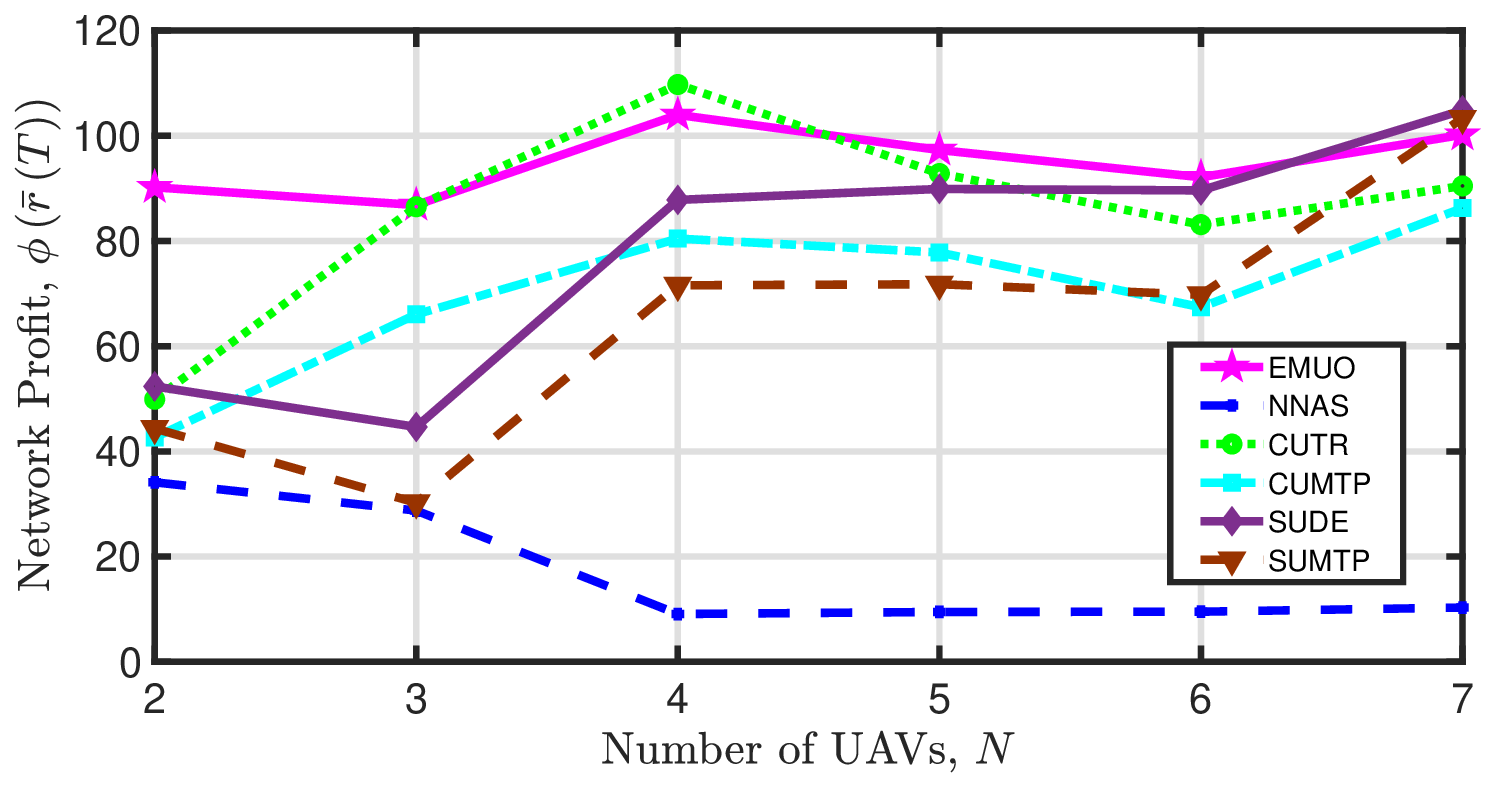}}
  \hspace{9pt}
    \subfigure[Network profit vs. number of subscribers]{
    \label{fig:subfig:b} 
    \includegraphics[width=3.2 in, height = 2.0 in]{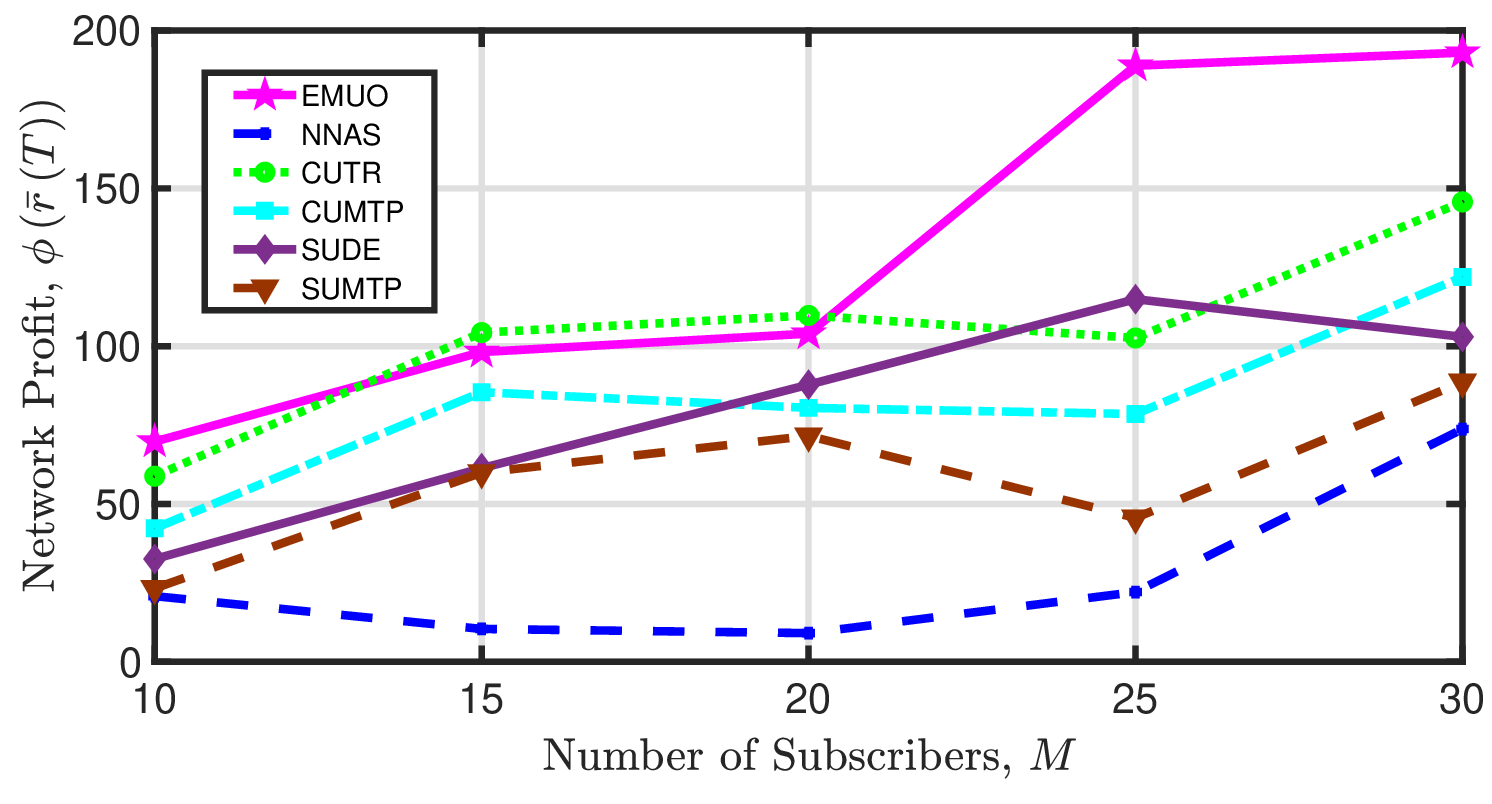}}
    \caption{Network profit of all comparison algorithms.}
\label{fig_network_profit}
\end{figure}

\textbf{Total latency:}
The tendency of the total latency of all subscribers versus various numbers of UAVs is illustrated in Fig. \ref{fig_total_latency}. 
From this figure we can observe that: 
1) The proposed EMUO algorithm achieves the lowest total latency. 
2) For all comparison algorithms except NNAS, more UAVs will lead to the reduction in total latency. This is because more subscribers can simultaneously receive video streams in a time slot when deploying more UAVs. For NNAS, it undermines the fairness of video transmission, and many subscribers have to wait for a long time to receive video streams from UAVs. 
3) When there are more subscribers, the achieved total latency of all comparison algorithms will increase as subscribers have to wait to be served. 
4) The achieved latency of CUTR and SUDE cannot be significantly reduced when increasing the number of UAVs. Although more subscribers can be simultaneously served, signal interference grows stronger and thus achievable bitrates of subscribers decrease. 
5) The achieved latency of EMUO can be reduced by at least 2.97\% compared with other algorithms. This indicates that the adverse impact of signal interference can be effectively alleviated through joint serving UAV selection optimization, UAV trajectory optimization, and UAV transmit power optimization. 
\begin{figure}[!t]
\centering
  \subfigure[Total latency vs. number of UAVs]{
    \label{fig:subfig:a} 
    \includegraphics[width=3.2 in, height = 2.0 in]{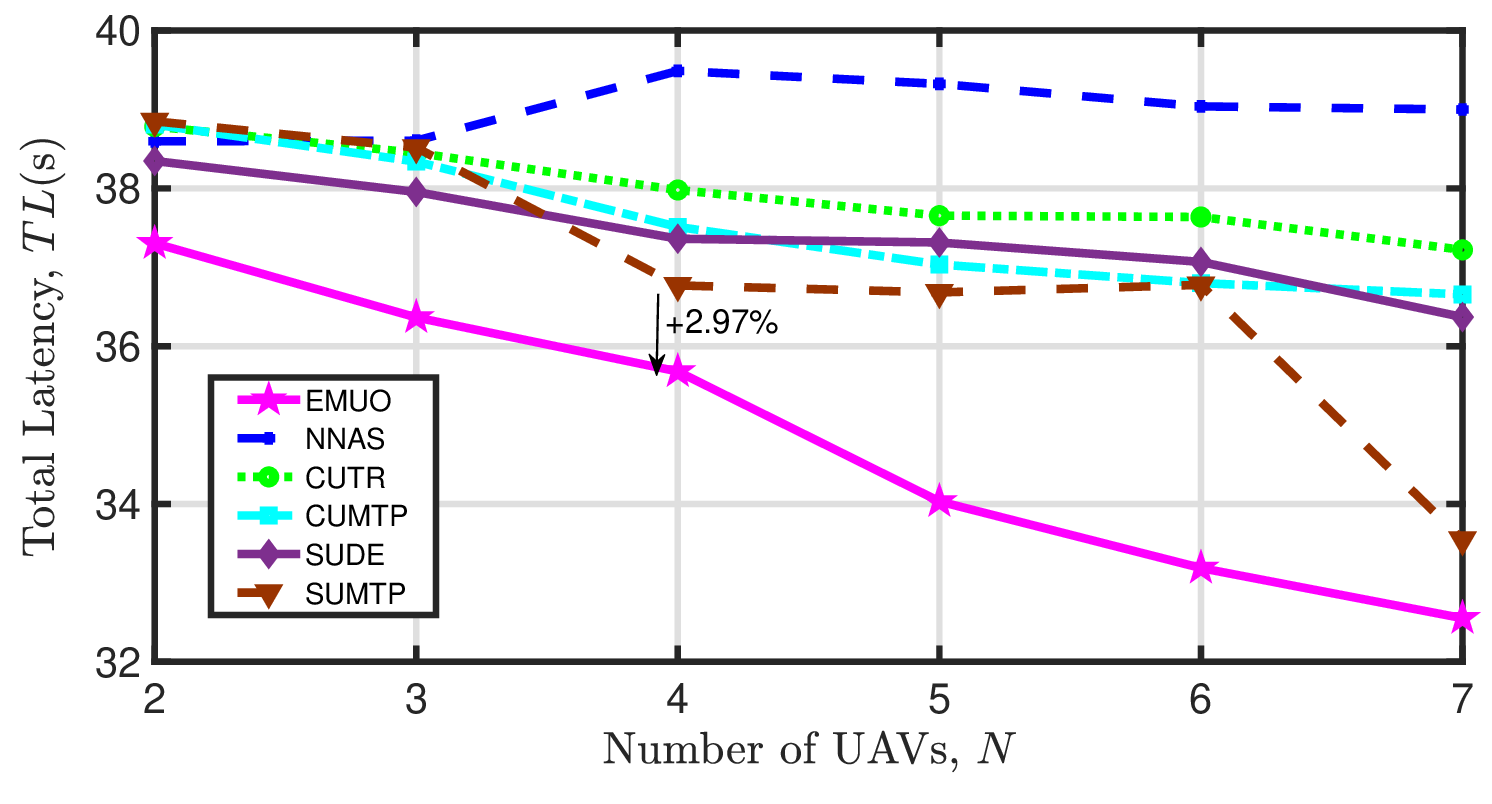}}
  \hspace{9pt}
    \subfigure[Total latency vs. number of subscribers]{
    \label{fig:subfig:b} 
    \includegraphics[width=3.2 in, height = 2.0 in]{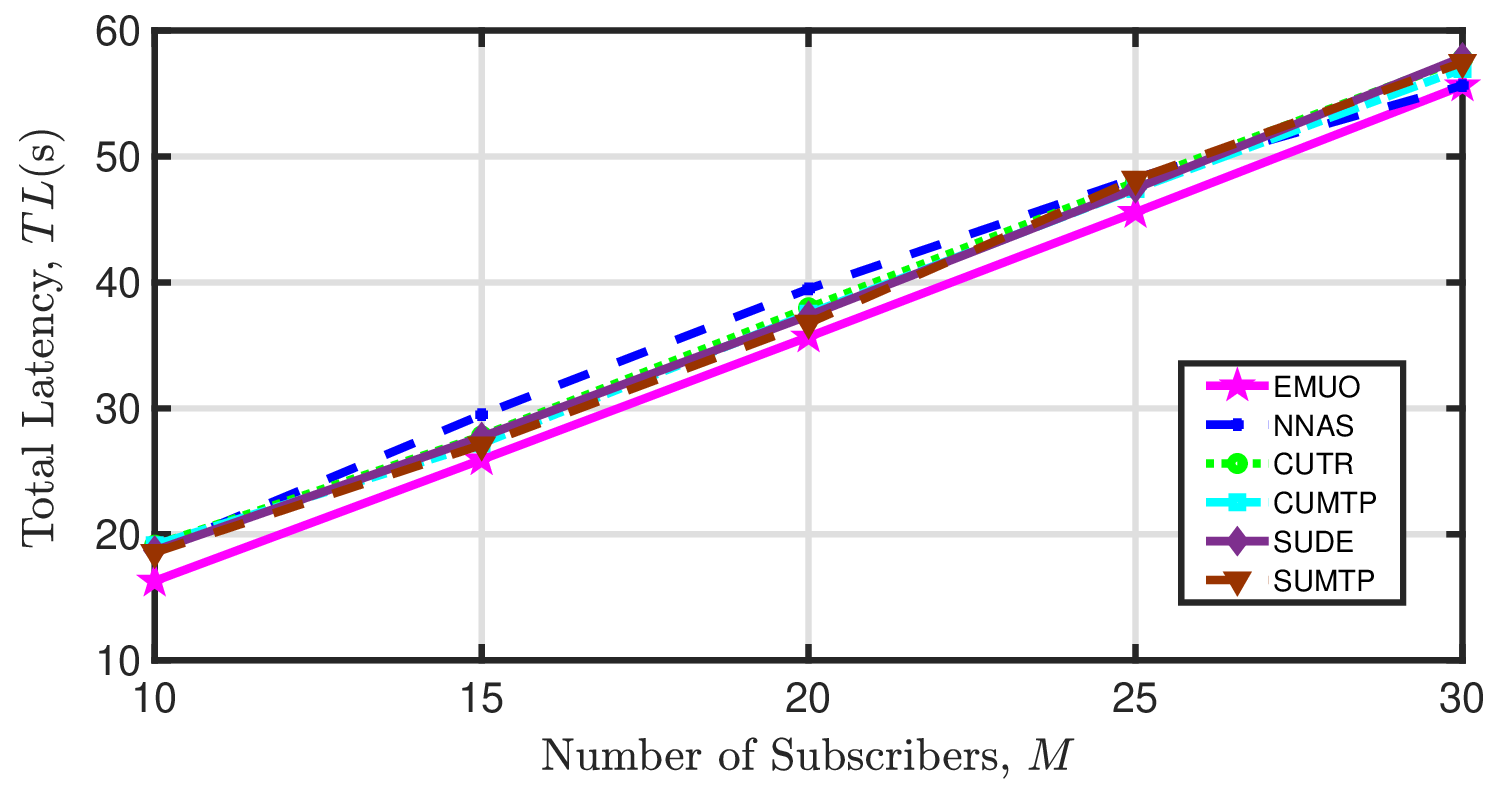}}
    \caption{Total latency of all comparison algorithms.}
\label{fig_total_latency}
\end{figure}

\textbf{QoE:} Besides, we depict the achieved QoE of all comparison algorithms when varying the number of UAVs and subscribers in Fig. \ref{fig_QoE}. 
The following observations can be obtained from this figure: 
1) The proposed EMUO algorithm outperforms the other comparison algorithms and achieves the largest QoE. When increasing the number of UAVs, the achieved QoE of EMUO increases mainly due to the rapid decrease in the total latency. 
2) For NNAS, as it obtains the smallest network profit and higher latency, its achieved QoE is the smallest. The large performance difference between NNAS and EMUO verifies the significance of performing serving UAV selection optimization. 
3) The obtained QoE of SUDE and SUMTP decrease with an increasing number of subscribers mainly due to the rapid increase in the generated latency. Additionally, for the two stationary UAV-based algorithms, UAVs are randomly and uniformly deployed and hover over the considered area. In this case, although the deployment locations of UAVs remain unchanged, more subscribers can be covered and served when the number of UAVs increases; and thus, higher QoE is obtained. 
4) The obtained QoE of both CUTR and CUMTP decrease with an increasing number of subscribers. For CUTR, its achieved QoE, however, shows an oscillation caused by special UAV trajectories and irregular signal interference when increasing the number of UAVs. 
Besides, when the number of UAVs exceeds four, the obtained QoE of CUMTP cannot be effectively improved. Recalling the results in Fig. \ref{fig_network_profit}(a), we can say that circular UAV trajectory-based algorithms cannot obtain high and stable QoE when the number of UAVs is large. 
\begin{figure}[!t]
\centering
  \subfigure[QoE vs. number of UAVs]{
    \label{fig:subfig:a} 
    \includegraphics[width=3.2 in, height = 2.0 in]{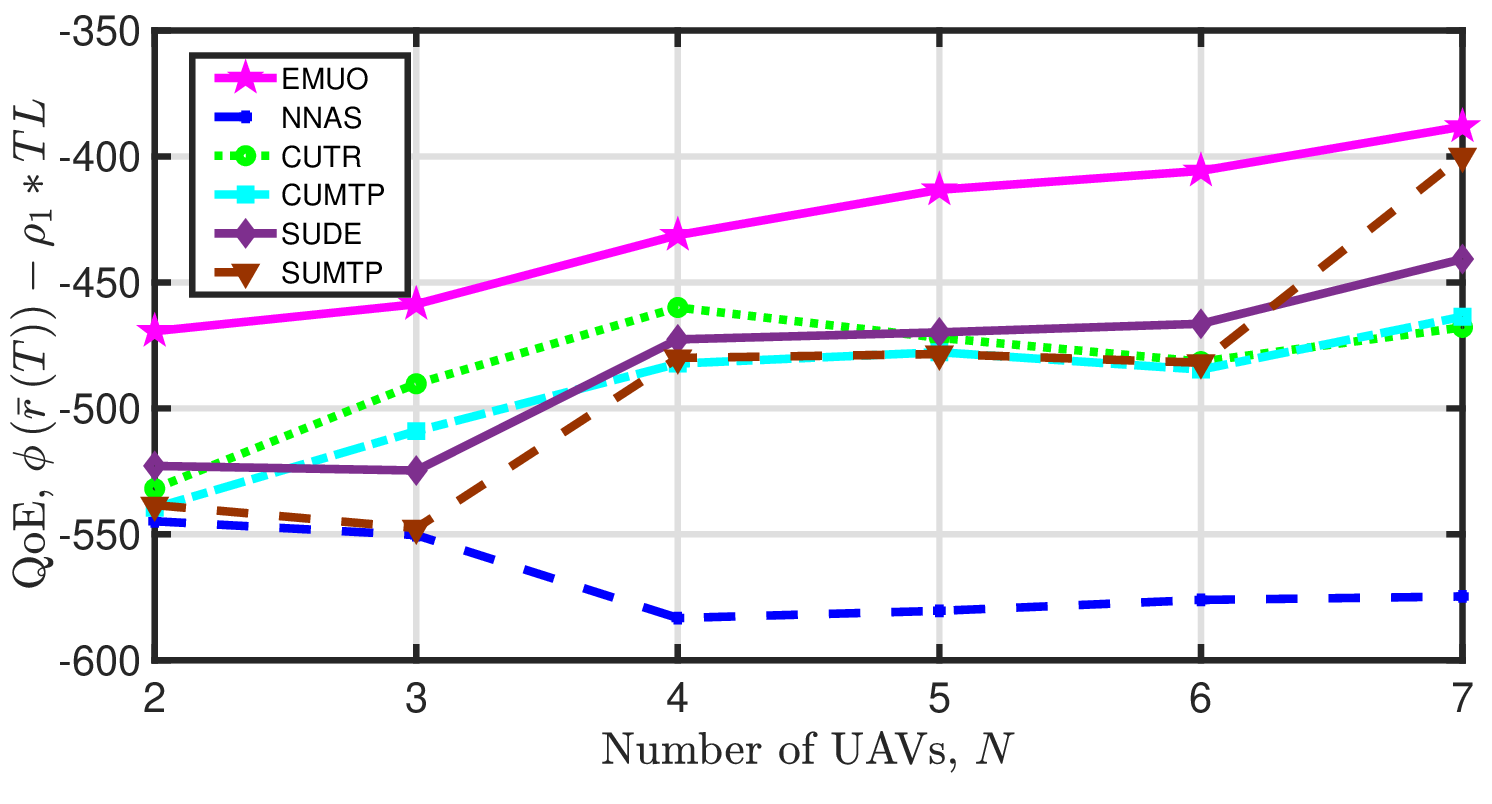}}
   \hspace{9pt}
    \subfigure[QoE vs. number of subscribers]{
    \label{fig:subfig:b} 
    \includegraphics[width=3.2 in, height = 2.0 in]{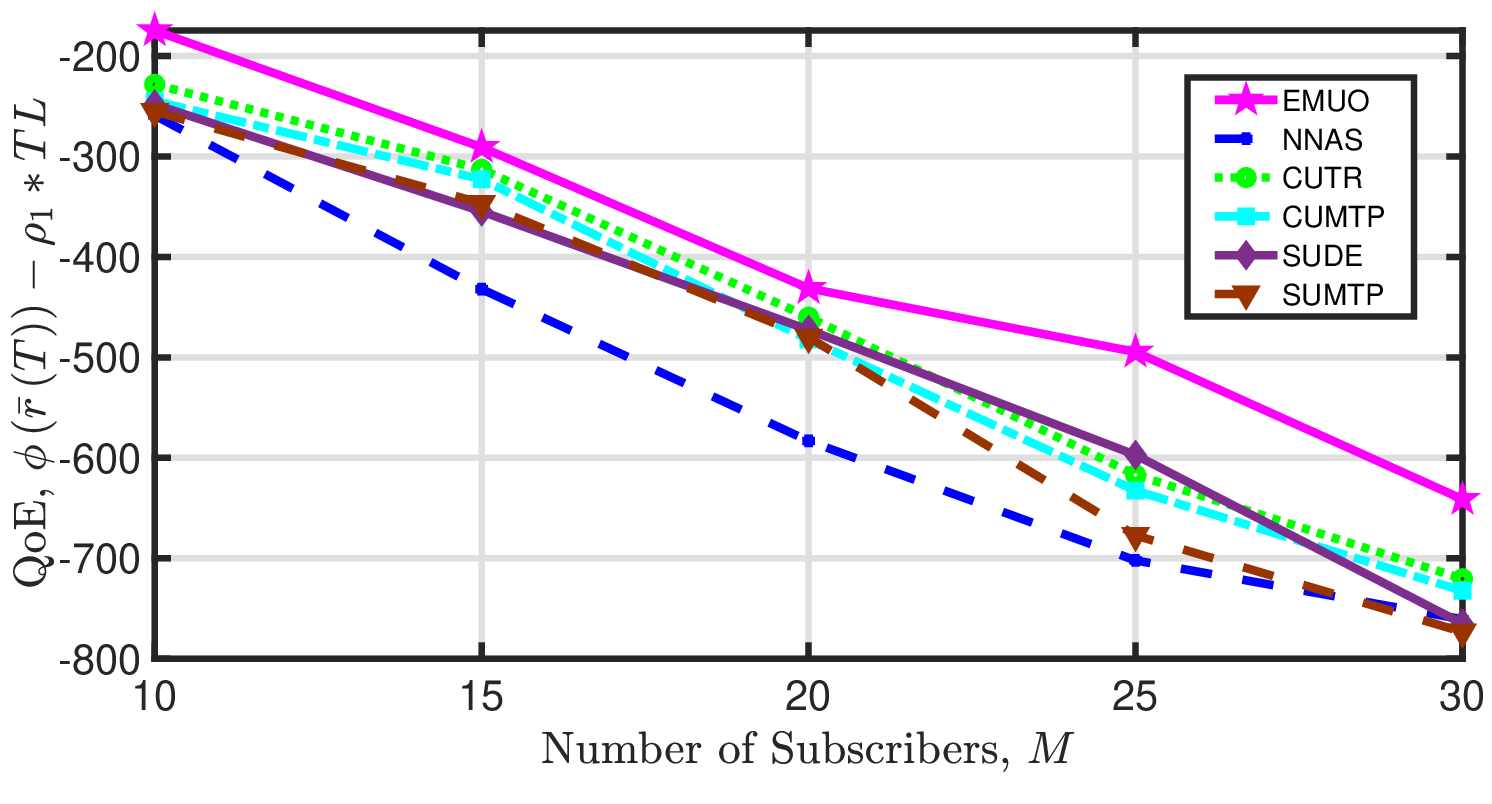}}
    \caption{QoE of all comparison algorithms.}
\label{fig_QoE}
\end{figure}

\subsubsection{Results of energy efficiency}
Similarly, we compare the proposed algorithm with all other comparison algorithms to evaluate its performance in terms of energy efficiency, which is defined as $EE = QoE - \rho_2 * TP$, where $TP =\sum\nolimits_{k \in {\mathcal K}} {\bar p_k^{\rm tot}\left( T \right)} $ represents the total power consumption of all UAVs.

First, we plot the total power consumption of the multi-UAV network by implementing different algorithms in Fig. \ref{fig_total_power_consumption}. 
The following observations can be obtained from this figure:
1) The two maximum-power-based algorithms, i.e., CUMTP and SUMTP, will consume the largest total power. 
2) When $N \ge 3$, the power consumption of all comparison algorithms will increase with an increasing number of UAVs as each UAV will consume circuit and transmit power to serve subscribers. 
3) Both signal interference and channel fading will affect the transmit power consumption. When $N \ge 3$, NNAS consumes the lowest transmit power. This fact indicates that power consumption due to anti-channel fading dominates the transmit power consumption. From the perspective of optimizing power consumption, the above fact further justifies the necessity of enabling LoS propagation for video transmission. 
4) The total power consumption of the two trajectory-optimized algorithms, i.e., EMUO and NNAS, decrease with an increasing $M$ when $M < 30$. This is because the increase in the number of subscribers leads to the reduction in the distances among UAVs and subscribers and then the weakening of channel fading. When there are more subscribers (e.g., $M=30$), NNAS suggests to transmit videos using great power to improve subscribers’ achievable bitrates and reduce latency, as shown in Fig. \ref{fig_total_latency}(b). 
\begin{figure}[!t]
\centering
  \subfigure[Total power consumption vs. number of UAVs]{
    \label{fig:subfig:a} 
    \includegraphics[width=3.2 in, height = 2.0 in]{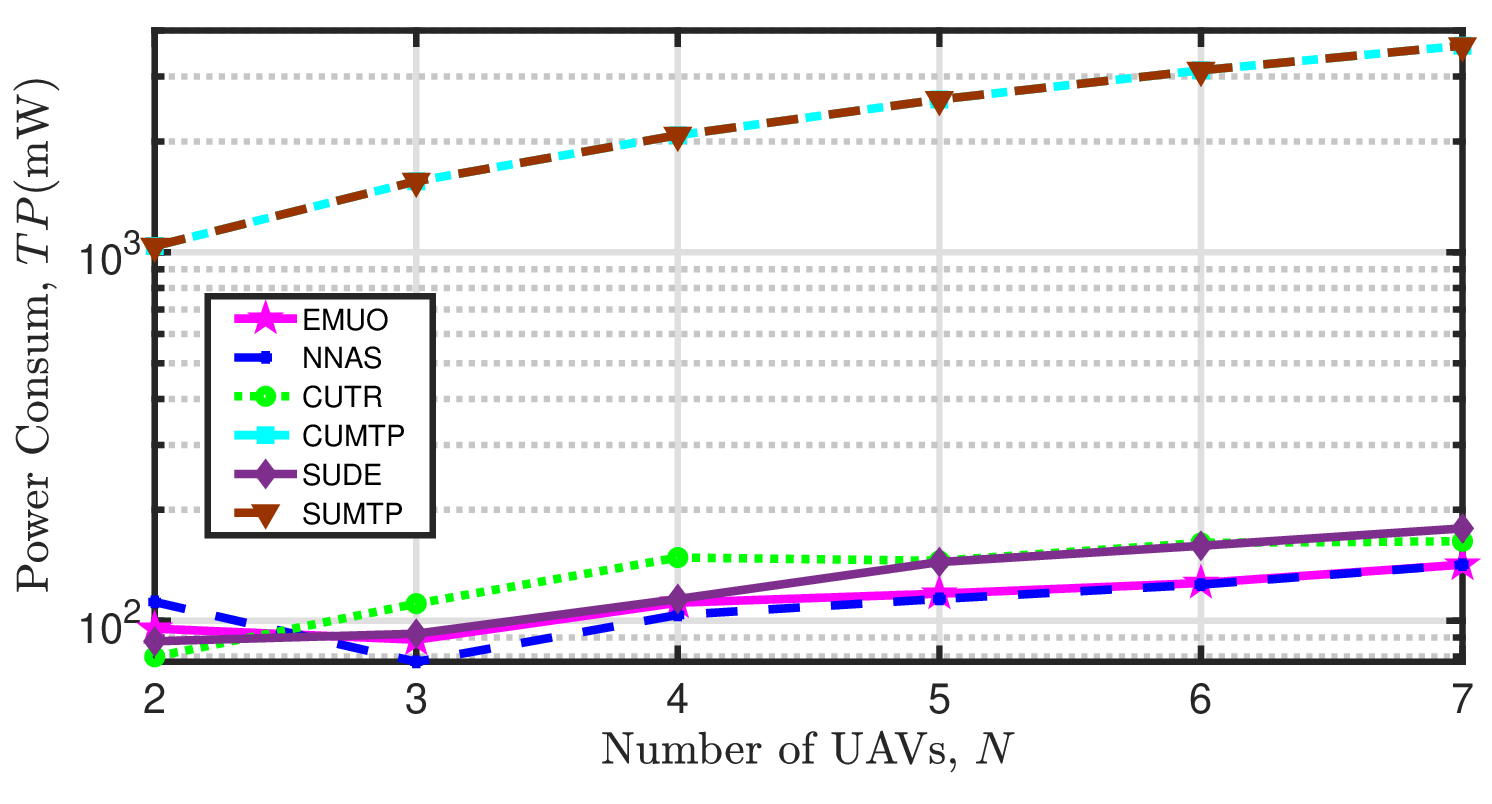}}
  \hspace{9pt}
    \subfigure[Total power consumption vs. number of subscribers]{
    \label{fig:subfig:b} 
    \includegraphics[width=3.2 in, height = 2.0 in]{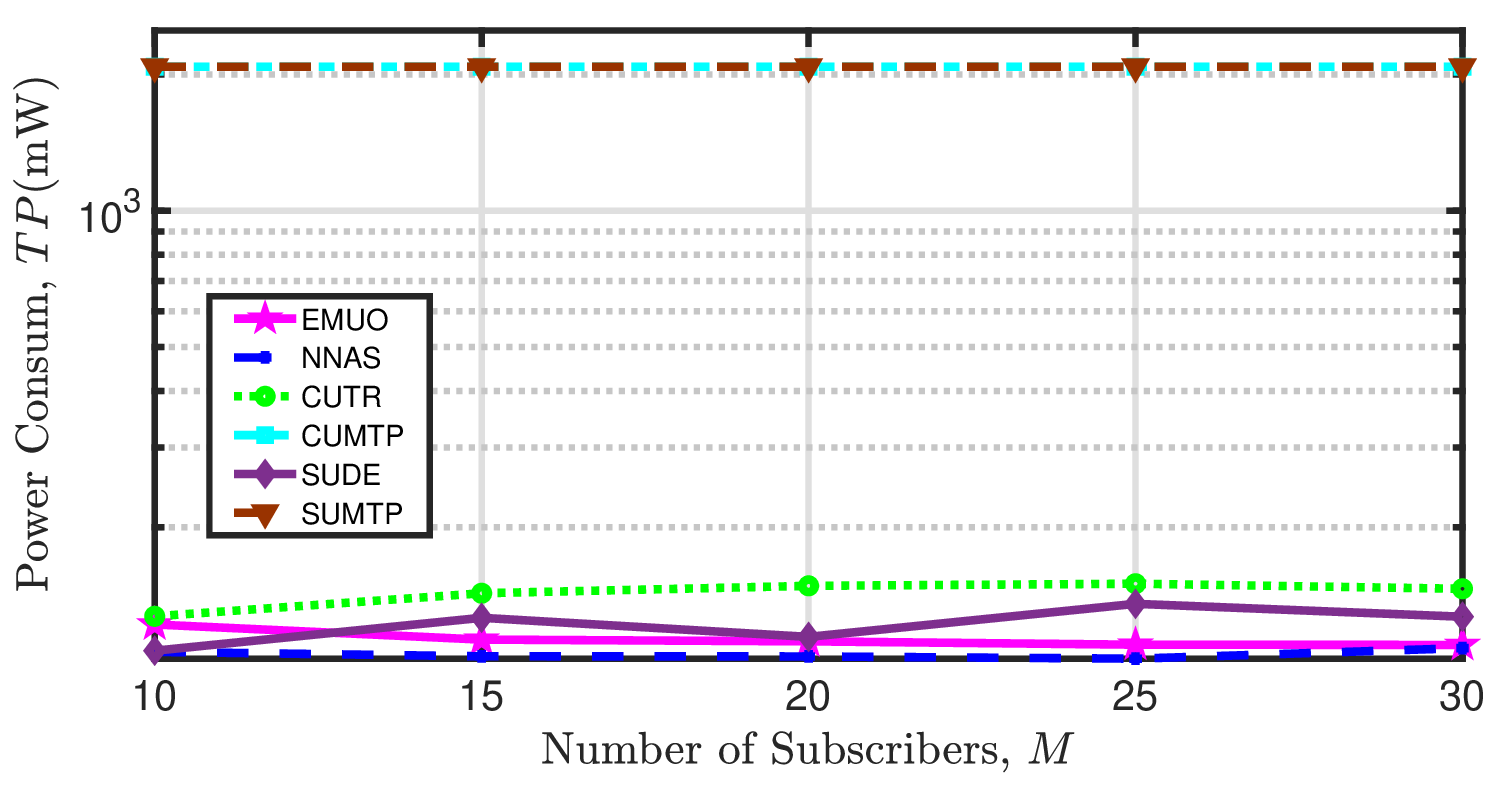}}
    \caption{Total power consumption of all comparison algorithms.}
\label{fig_total_power_consumption}
\end{figure}

Second, the influence of the number of UAVs and the number of subscribers on the achieved energy efficiency of all comparison algorithms is depicted in Fig. \ref{fig_Energy_efficiency}. 
From this figure, we can have the following observations: 
1) Given any UAV number, the  energy efficiency achieved by all comparison algorithms decrease with an increasing number of subscribers mainly owing to the rapid reduction in their achieved QoE. 
2) When the number of UAVs equals two and the number of subscribers is fifteen, the achieved energy efficiency of EMUO is slightly smaller than that of CUTR. 
3) Except for the above case, EMUO is the most energy-efficient algorithm. For instance, the energy efficiency of the multi-UAV network can be improved by up to $66.75\%$ by running EMUO. 
4) Except for the two maximum-power-based algorithms, NNAS obtains the smallest energy efficiency. It is interesting to observe that CUTR is energy-efficient than SUDE when $N \le 3$ and $M \le 15$, and the performance of the two comparison algorithms is reversed when $N \ge 4$ and $M \ge 20$. In this case, its achieved energy efficiency can be greater than that of the proposed EMUO. Therefore, the above results recommend the following UAV trajectory design scheme: considering the simple UAV movement control, it may be a good choice to deploy UAVs with circular trajectories when the number of UAVs is small and subscribers are sparsely distributed in the considered area.
\begin{figure}[!t]
\centering
\includegraphics[width=3.2 in, height = 1.6 in]{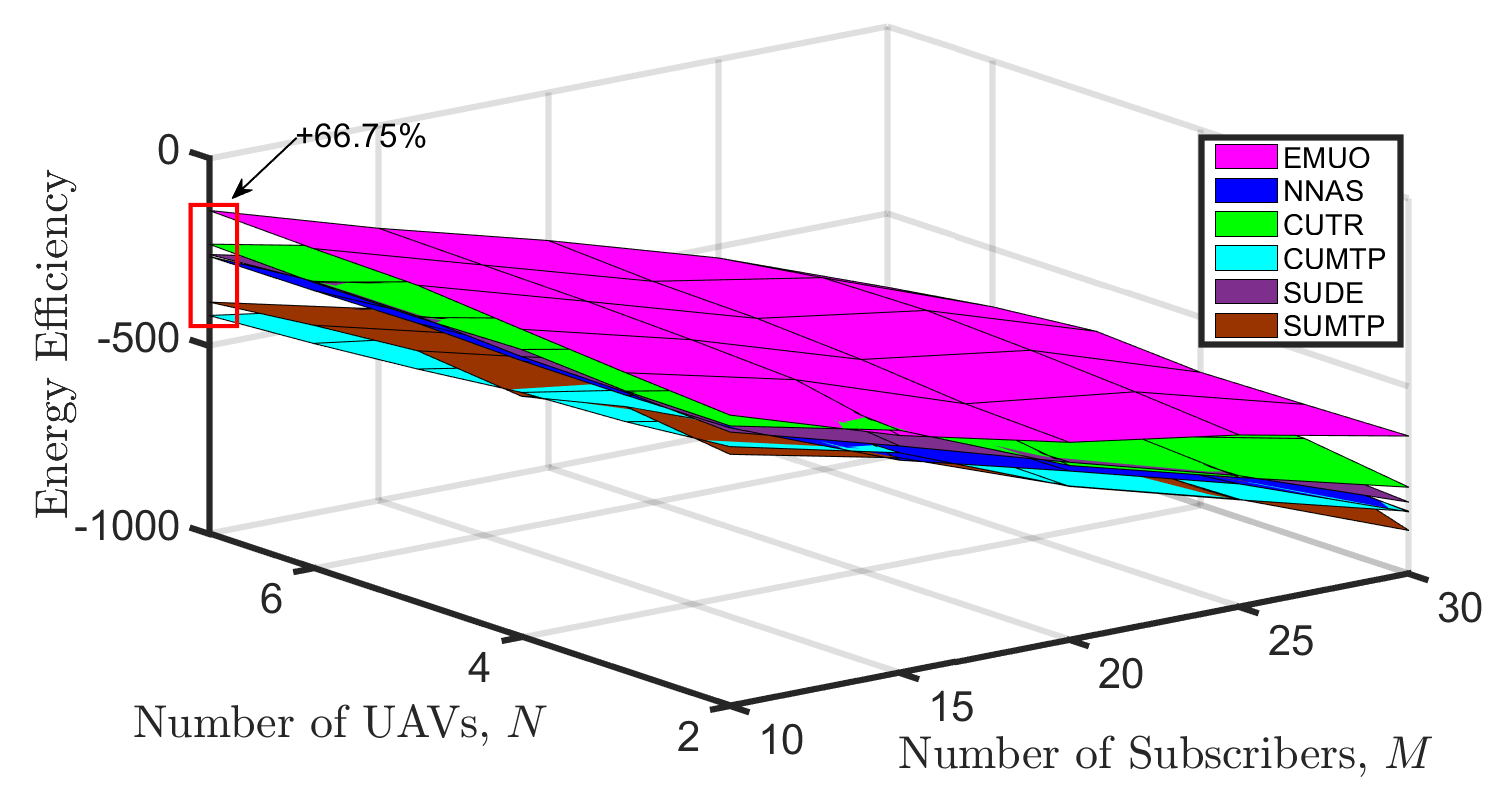}
\caption{Energy efficiency vs. the number of subscribers and UAVs.}
\label{fig_Energy_efficiency}
\end{figure}
{
\subsubsection{Results of fairness}
Finally, we verify the fairness of the proposed algorithm by comparing all algorithms. To this end, the popular Jain’s fairness index is taken as the evaluation indicator, which can be calculated by ${{RF = {{\left( {\sum\limits_{i = 1}^M {\frac{{{{\bar r}_i}\left( T \right)}}{{{R_i}}}} } \right)}^2}} \mathord{\left/ {\vphantom {{RF = {{\left( {\sum\limits_{i = 1}^M {\frac{{{{\bar r}_i}\left( T \right)}}{{{R_i}}}} } \right)}^2}} {M{{\sum\limits_{i = 1}^M {\left( {\frac{{{{\bar r}_i}\left( T \right)}}{{{R_i}}}} \right)} }^2}}}} \right. \kern-\nulldelimiterspace} {M{{\sum\limits_{i = 1}^M {\left( {\frac{{{{\bar r}_i}\left( T \right)}}{{{R_i}}}} \right)} }^2}}}$.
\begin{figure}[!t]
\centering
\includegraphics[width=3.2 in, height = 1.6 in]{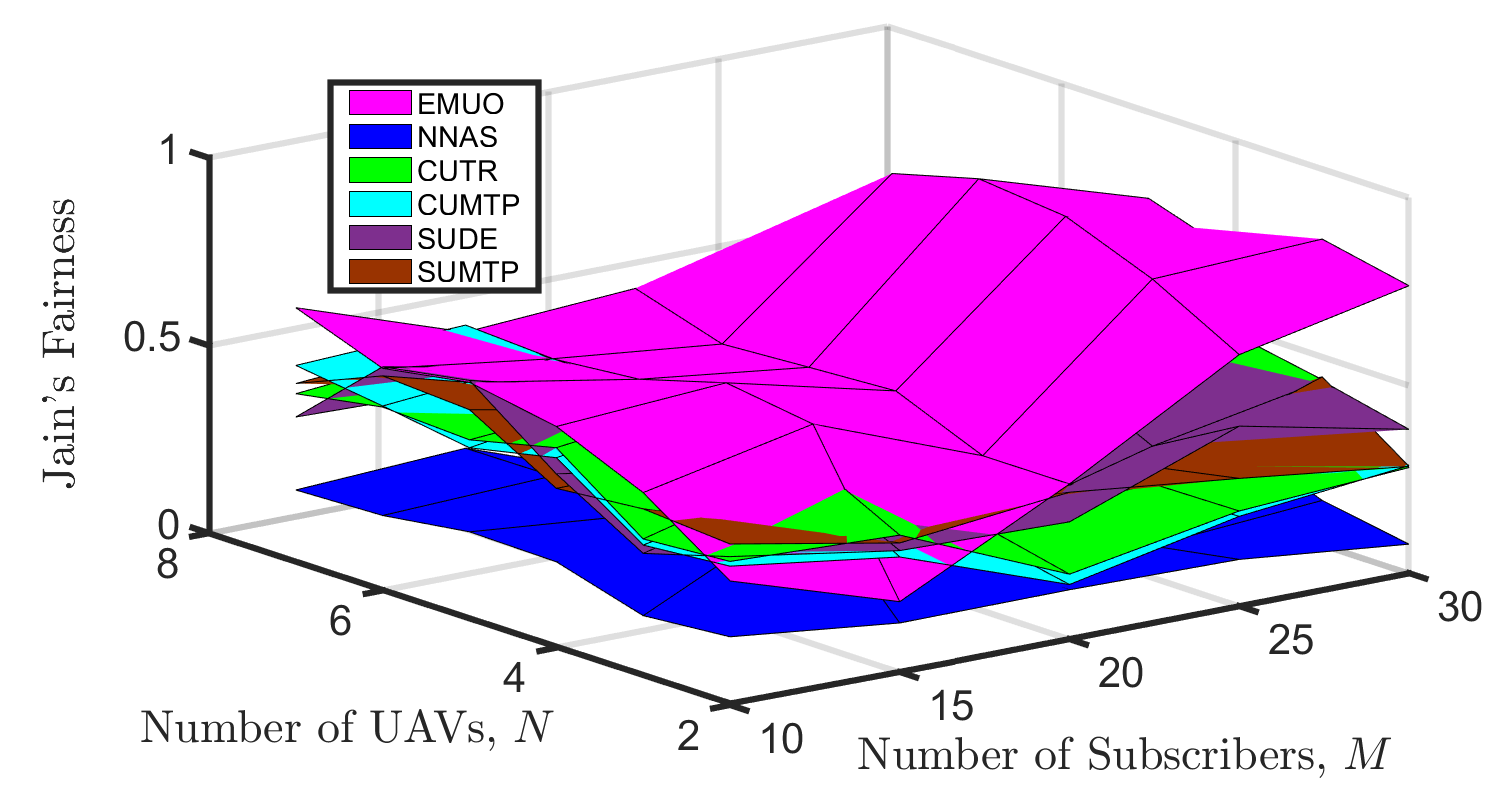}
\caption{{Jain's fairness vs. the number of subscribers and UAVs.}}
\label{fig_fairness}
\end{figure}
The achieved fairness indexes of all comparison algorithms when varying the number of UAVs and the number of subscribers are shown in Fig. \ref{fig_fairness}. 
The following observations can be obtained from this figure:
1) Compared with other algorithms, the multi-UAV network can provide fairer video transmission services for subscribers by running EMUO when $N >3$ and $M >15$. When $N \le 3$ and $M \le 15$, EMUO is overwhelmed by CUTR regarding service fairness. 
2) As explained above, the option of maximizing the UAV transmit power will not result in higher service fairness. 
3) Likewise, the achieved Jain’s fairness index of CUTR is larger than that of SUDE when $N \le 3$ and $M \le 15$, and SUDE can achieve a higher fairness value than CUTR when $N \ge 4$ and $M \ge 20$. 4) NNAS undermines the fairness brought by the UAV trajectory optimization and obtains the smallest fairness value. 
}

In summary, when there are few UAVs and subscribers are sparsely distributed, deploying UAVs with circular trajectories can be an acceptable proposal to achieve energy-efficient video transmission. On the contrary, when there are more UAVs or subscribers, a joint UAV trajectory, transmit power, and serving UAV selection optimization scheme for video transmission will be more energy-efficient.


\section{Conclusion}
In this paper, we investigated the issue of energy-efficient multi-UAV network optimization for subscribers' QoE-driven video transmission. In particular, considering the time-varying characteristic of the multi-UAV network and the differentiated requirements of subscribers, we first designed a novel QoE model simultaneously considering the video bitrate, the latency, and the frame freezing. 
Next, we formulated the UAV video transmission problem as a sequential-decision problem to maximize the QoE of subscribers and minimize the total power consumption of the multi-UAV network through the joint optimization and control of UAV transmit power, serving UAV selection, and UAV trajectories. 
Further, we developed an energy-efficient multi-UAV network optimization algorithm with provable performance guarantees to solve the challenging problem by introducing the key ideas of Lyapunov optimization, iterative optimization, and Taylor expansions. 
Finally, extensive simulation results showed that the proposed algorithm can effectively improve the QoE of subscribers and is 66.75\% more energy-efficient than benchmarks.
Although we create a QoE model incorporating the video bitrate, the latency, and the
frame freezing, the modeling of subscriber’s playback buffer state is not investigated in this paper. How to integrate the buffer state into the QoE model is challenging and deserves further research in the near future.


%

\bibliographystyle{IEEEtran}
\bibliography{Globecom_RAN_slicing}

\newpage

\appendices
\section{Proof of Proposition \ref{propo_1}}
Firstly, we can rewrite the LHS expression of (\ref{eq:transmission_power_transform_problem}b) as
\begin{equation}\label{eq:transmission_power_constraint_problem}
	\begin{array}{l}
		\sum\limits_{k \in {\mathcal K}} {{c_{ik}}\left( t \right){{\log }_2}\left( {\frac{{{\sigma ^2} + {p_k}\left( t \right){h_{ik}}\left( t \right){\rm{ + }}\sum\limits_{j \in {\mathcal K}\backslash \left\{ k \right\}} {{p_j}\left( t \right){h_{ij}}\left( t \right)} }}{{{\sigma ^2}{\rm{ + }}\sum\limits_{j \in {\mathcal K}\backslash \left\{ k \right\}} {{p_j}\left( t \right){h_{ij}}\left( t \right)} }}} \right)}  \\
		= \sum\limits_{k \in {\mathcal K}} {{c_{ik}}\left( t \right)\left\{ {\underbrace {{{\log }_2}\left( {{\sigma ^2} + \sum\limits_{j \in {\mathcal K}} {{p_j}\left( t \right){h_{ij}}\left( t \right)} } \right)}_{\rm{concave \text{ } function}}} \right.} \\
		\left. {\underbrace { - {{\log }_2}\left( {{\sigma ^2}{\rm{ + }}\sum\limits_{j \in {\mathcal K}\backslash \left\{ k \right\}} {{p_j}\left( t \right){h_{ij}}\left( t \right)} } \right)}_{\rm{convex \text{ } function}}} \right\}.
	\end{array}
\end{equation}

(\ref{eq:transmission_power_constraint_problem}) is the sum of a concave function and a convex function and is non-concave. We then explore a successive convex approximate (SCA) method to handle the non-convexity of the corresponding constraint (\ref{eq:transmission_power_transform_problem}b). The core idea of SCA is to obtain a local approximate solution of a non-convex original problem by iteratively solving a series of convex optimization problems  with different initial points. Since these local convex problems have the similar geometric characteristics as the non-convex problem, they can be regarded as the approximation of the non-convex problem at the given local points. 

Next, we discuss how to obtain the approximation function of non-concave function. Based on the following assumption, we can obtain the approximation function. 
\begin{assumpt}\label{assumption_1}
	\rm A function $\tilde f:{\mathcal X} \times {\mathcal Y} \to {\mathbb R}$ is assumed to be the approximation
	function of the non-concave function $f(x)$ ($x \in {\mathcal X}$), when the following conditions hold \cite{scutari2017parallel}
	\begin{itemize}
		\item $\tilde f(\cdot, \cdot)$ is continuous in ${\mathcal X} \times {\mathcal Y}$.
		\item $\tilde f(\cdot, x^{(r)})$ is concave in ${\mathcal X}$ for all $x^{(r)} \in {\mathcal Y}$.
		\item Function value consistency: $\tilde f(x, x) = f(x)$ for all $x \in {\mathcal X}$.
		\item Gradient consistency: $\frac{{\partial \tilde f(x,{x^{(r)}})}}{{\partial x}}{|_{x = {x^{(r)}}}} = \nabla f(x){|_{x = {x^{(r)}}}}$ for all $x^{(r)} \in {\mathcal Y}$.
		\item Lower bound: $f(x) \ge \tilde f(x, x^{(r)})$ for all $x \in {\mathcal X}$, $x^{(r)} \in {\mathcal Y}$.
	\end{itemize}
\end{assumpt}

In other words, we can approximate an original function with its lower-bound that has the same first-order behavior at each iteration. This is what the SCA method does.

Based on the principle of SCA method, let ${\mathord{\buildrel{\lower3pt\hbox{$\scriptscriptstyle\smile$}} 
		\over \Lambda } _{ik}}\left( t \right) =  - {\log _2}\left( {{\sigma ^2}{\rm{ + }}\sum\limits_{j \in {\mathcal K}\backslash \left\{ k \right\}} {{p_j}\left( t \right){h_{ij}}\left( t \right)} } \right)$, and for any local point $p_j^{\left( r \right)}\left( t \right)$, which denotes the transmit power of the $j$-th UAV at the $r$-th iteration when exploring the SCA method, we can obtain the lower bound of ${\mathord{\buildrel{\lower3pt\hbox{$\scriptscriptstyle\smile$}} 
		\over \Lambda } _{ik}}\left( t \right)$  by the first-order Taylor expansion, i.e., 
\begin{equation}\label{eq:transmission_power_sca}
	\begin{array}{l}
		{{\mathord{\buildrel{\lower3pt\hbox{$\scriptscriptstyle\smile$}} 
					\over \Lambda } }_{ik}}\left( t \right) \ge  - {\log _2}\left( {{\sigma ^2}{\rm{ + }}\sum\limits_{j \in {\mathcal K}\backslash \left\{ k \right\}} {p_j^{\left( r \right)}\left( t \right){h_{ij}}\left( t \right)} } \right)  \\
		- \sum\limits_{j \in {\mathcal K}\backslash \left\{ k \right\}} {\frac{{{h_{ij}}\left( t \right)}}{{\left( {{\sigma ^2}{\rm{ + }}\sum\limits_{j \in {\mathcal K}\backslash \left\{ k \right\}} {p_j^{\left( r \right)}\left( t \right){h_{ij}}\left( t \right)} } \right)\ln 2}}\left( {{p_j}\left( t \right) - p_j^{\left( r \right)}\left( t \right)} \right)} \\
		=  - F_{ik}^{\left( r \right)}\left( t \right) - \sum\limits_{j \in {\mathcal K}\backslash \left\{ k \right\}} {G_{ik}^{\left( r \right)}\left( t \right)\left( {{p_j}\left( t \right) - p_j^{\left( r \right)}\left( t \right)} \right)} 
	\end{array}
\end{equation}

Consequently, the constraint (\ref{eq:transmission_power_transform_problem}b) can be {approximated} as (\ref{eq:transmission_power_constraint_simplify}).
This completes the proof.

\section{Proof of Lemma \ref{lemma:lemma_transmission_power_optimization}}
Without loss of generality, we assume that the video access service of subscriber $i$ is relayed by a UAV, (\ref{eq:transmission_power_constraint_simplify}) can then be simplified as 
\begin{equation}\label{eq:transmission_power_mosek_b}
	\begin{array}{l}
		{{\hat \Lambda }_{i}}\left( t \right) - F_{ik}^{\left( r \right)}\left( t \right) - \sum\limits_{j \in {\mathcal K}\backslash \left\{ k \right\}} {G_{ik}^{\left( r \right)}\left( t \right)\left( {{p_j}\left( t \right) - p_j^{\left( r \right)}\left( t \right)} \right)} \\
		\ge {\eta _i}\left( t \right),\forall i, t
	\end{array}
\end{equation}

Considering that ${\hat \Lambda _{i}}\left( t \right) = {\log _2}\left( {{\sigma ^2} + \sum\limits_{j \in {\mathcal K}} {{p_j}\left( t \right){h_{ij}}\left( t \right)} } \right)$ is a complex logarithmic function, we then introduce a family of auxiliary variables to transform (\ref{eq:transmission_power_mosek_b})
into convex cones. In particular, (\ref{eq:transmission_power_mosek_b}) can be rewritten as
\begin{equation}\label{eq:transmission_power_mosek_b_exponent}
	\begin{array}{l}
		{\sigma ^2} + \sum\limits_{j \in {\mathcal K}} {{p_j}\left( t \right){h_{ij}}\left( t \right)}  \ge \\
		{e^{\left( {{\eta _i}\left( t \right) + F_{ik}^{\left( r \right)}\left( t \right) + \sum\limits_{j \in {\mathcal K}\backslash \left\{ k \right\}} {G_{ik}^{\left( r \right)}\left( t \right)\left( {{p_j}\left( t \right) - p_j^{\left( r \right)}\left( t \right)} \right)} } \right)\ln 2}}
	\end{array}
\end{equation}

To simplify this expression, we introduce three auxiliary variables $Z_i^1$, $Z_i^2$, $Z_i^3$ ($\forall i \in {\mathcal{I}}$), and let 
\begin{subequations}\label{eq:transmission_power_mosek}
	\begin{alignat}{2}
		&Z_i^1 = {\sigma ^2} + \sum\limits_{j \in {\mathcal K}} {{p_j}\left( t \right){h_{ij}}\left( t \right)} \\
		&Z_i^2 = 1\\
		&Z_i^3 = \left( {{\eta _i}\left( t \right) + F_{ik}^{\left( r \right)}\left( t \right) + } \right.\notag\\
		&\left. {\sum\limits_{j \in {\mathcal K}\backslash \left\{ k \right\}} {G_{ik}^{\left( r \right)}\left( t \right)\left( {{p_j}\left( t \right) - p_j^{\left( r \right)}\left( t \right)} \right)} } \right)\ln 2
	\end{alignat}
\end{subequations}

It turns out that (\ref{eq:transmission_power_mosek}a)-(\ref{eq:transmission_power_mosek}c) are linear constraints. 
According to the standard form of an exponential cone ${K_{\exp }} = \left\{ {x \in {R^3}:{x_1} \ge {x_2}\exp ({{{x_3}} \mathord{\left/
			{\vphantom {{{x_3}} {{x_2}}}} \right.
			\kern-\nulldelimiterspace} {{x_2}}}),{x_1},{x_2} \ge 0} \right\}$, (\ref{eq:transmission_power_mosek_b_exponent}) can take a different form
\begin{equation}\label{eq:transmission_power_mosek_b_exponential_cone}
	\begin{array}{l}
		\left( {Z_i^1,Z_i^2,Z_i^3} \right) \in {K_{\exp }}
	\end{array}
\end{equation}

Similarly, we can rewrite (\ref{eq:latency_related_transform_problem}b) into the following expression of a standard rotated quadratic cone 
\begin{equation}\label{eq:transmission_power_mosek_d_rotated_cone}
	\begin{array}{l}
		\left( {{\eta _i}\left( t \right),{\xi _i}\left( t \right),\sqrt {{{2L} \mathord{\left/
						{\vphantom {{2L} B}} \right.
						\kern-\nulldelimiterspace} B}} } \right) \in Q_r^3,\forall i
	\end{array}
\end{equation}
where, $Q_r^n = \left\{ {x \in {R^n}:2{x_1}{x_2} \ge \sum\limits_{j = 3}^n {x_j^2} ,{x_1} \ge 0,{x_2} \ge 0} \right\}$.

Based on {the approximation of (\ref{eq:transmission_power_transform_problem}b) as (\ref{eq:transmission_power_constraint_simplify})} and the above derivation, we can approximately transform {(\ref{eq:transmission_power_problem})} into the standard conic problem (\ref{eq:transmission_power_slack_problem_mosek}). 

{Further, the approximation result in Proposition \ref{propo_1} indicates that the feasible region of (\ref{eq:transmission_power_slack_problem_mosek}) is a subset of that of (\ref{eq:transmission_power_problem}). Therefore, the opposite value of the maximum value of (\ref{eq:transmission_power_slack_problem_mosek}a) is the upper bound of the optimal objective value of (\ref{eq:transmission_power_problem}).}
This completes the proof. 

\section{Proof of Proposition \ref{propo_2}}
By substituting (\ref{eq:h_ik_t_new}) into the constraint (\ref{eq:trajectories_slack_problem}b), we obtain
\begin{equation}\label{eq:trajectories_problem_constraint_transform}
	\begin{array}{l}
		\sum\limits_{k \in {\mathcal K}} {{c_{ik}}\left( t \right){{\log }_2}\left( {\frac{{{\sigma ^2} + {p_k}\left( t \right){h_{ik}}\left( t \right){\rm{ + }}\sum\limits_{j \in {\mathcal K}\backslash \left\{ k \right\}} {{p_j}\left( t \right){h_{ij}}\left( t \right)} }}{{{\sigma ^2}{\rm{ + }}\sum\limits_{j \in {\mathcal K}\backslash \left\{ k \right\}} {{p_j}\left( t \right){h_{ij}}\left( t \right)} }}} \right)} \\
		= \sum\limits_{k \in {\mathcal K}} {{c_{ik}}\left( t \right){{\log }_2}\left( {{\sigma ^2} + \sum\limits_{j \in {\mathcal K}} {\frac{{{p_j}\left( t \right){\omega _{ij}}}}{{{H^2} + {{\left\| {\bm{q_j} \left( t \right) - \bm{s_i}\left( t \right)} \right\|}^2}}}} } \right)} \\
		- \sum\limits_{k \in {\mathcal K}} {{c_{ik}}\left( t \right){{\log }_2}\left( {{\sigma ^2} + \sum\limits_{j \in {\mathcal K}\backslash \left\{ k \right\}} {\frac{{{p_j}\left( t \right){\omega _{ij}}}}{{{H^2} + {{\left\| {\bm{q_j} \left( t \right) - \bm{s_i}\left( t \right)} \right\|}^2}}}} } \right)} \\
		\ge {\eta _i}\left( t \right),\forall i,t
	\end{array}
\end{equation}

However, (\ref{eq:trajectories_problem_constraint_transform}) is not a convex constraint w.r.t ${\bm{q_j} \left( t \right)}$. Fortunately, ${\mathord{\buildrel{\lower3pt\hbox{$\scriptscriptstyle\smile$}} 
		\over \Lambda } _{ik}}\left( t \right) =  - {\log _2}\left( {{\sigma ^2} + \sum\limits_{j \in {\mathcal K}\backslash \left\{ k \right\}} {\frac{{{p_j}\left( t \right){\omega _{ij}}}}{{{H^2} + {{\left\| {\bm{q_j} \left( t \right) - \bm{s_i}\left( t \right)} \right\|}^2}}}} } \right)$ is concave w.r.t ${\left\| {\bm{q_j} \left( t \right) - \bm{s_i}\left( t \right)} \right\|^2}$. 
Therefore, a slack variable ${B_{ij}}\left( t \right) \le {\left\| {\bm{q_j} \left( t \right) - \bm{s_i}\left( t \right)} \right\|^2},\forall i,j \ne k,t$, can be introduced to alleviate the problem.
The LHS expression of (\ref{eq:trajectories_problem_constraint_transform}) can be rewritten as
\begin{equation}\label{eq:trajectories_problem_constraint_transform_slack}
	\begin{array}{*{20}{l}}
		{\underbrace {\sum\limits_{k \in {\cal K}} {{c_{ik}}\left( t \right){{\log }_2}\left( {{\sigma ^2} + \sum\limits_{j \in {\cal K}} {\frac{{{p_j}\left( t \right){\omega _{ij}}}}{{{H^2} + {{\left\| {{\bm q_j}\left( t \right) - {\rm{ }}{\bm s_i}\left( t \right)} \right\|}^2}}}} } \right)} }_{{\rm{convex \text{ } function}}}}\\
		{\underbrace { - \sum\limits_{k \in {\cal K}} {{c_{ik}}\left( t \right){{\log }_2}\left( {{\sigma ^2} + {\sum _{j \in {\cal K}\backslash \left\{ k \right\}}}\frac{{{p_j}\left( t \right){\omega _{ij}}}}{{{H^2} + {B_{ij}}\left( t \right)}}} \right)} }_{{\rm{concave \text{ } function}}}}
	\end{array}
\end{equation}

Observe that the above expression is a sum of a convex function and a concave function. 
To tackle the non-convex constraint, we explore the SCA method. 
Let ${\hat \Lambda _{i}}\left( t \right) = {\log _2}\left( {{\sigma ^2} + \sum\limits_{j \in {\mathcal K}} {\frac{{{p_j}\left( t \right){\omega _{ij}}}}{{{H^2} + {{\left\| {\bm{q_j} \left( t \right) - \bm{s_i}\left( t \right)} \right\|}^2}}}} } \right)$.
For any given local point $\bm{q_j}^{\left( r \right)}\left( t \right)$, which denotes the 2D horizontal location of the $j$-th UAV at the $r$-th iteration of the SCA method, we can obtain the lower bound of ${\hat \Lambda _{i}}\left( t \right)$ via conducting the first-order Taylor expansion, i.e., 
\begin{equation}\label{eq:trajectories_problem_constraint_sca}
	\begin{array}{l}
		{{\hat \Lambda }_{i}}\left( t \right) \ge {\log _2}\left( {{\sigma ^2} + \sum\limits_{j \in {\mathcal K}} {\frac{{{p_j}\left( t \right){\omega _{ij}}}}{{{H^2} + {{\left\| {q_j^{\left( r \right)}\left( t \right) - \bm{s_i}\left( t \right)} \right\|}^2}}}} } \right)\\
		- \sum\limits_{j \in {\mathcal K}} {\frac{{\frac{{{p_j}\left( t \right){\omega _{ij}}}}{{{{\left( {{H^2} + {{\left\| {\bm{q_j}^{\left( r \right)}\left( t \right) - \bm{s_i}\left( t \right)} \right\|}^2}} \right)}^2}}}\left( {{{\left\| {\bm{q_j} \left( t \right) - \bm{s_i}\left( t \right)} \right\|}^2} - {{\left\| {\bm{q_j}^{\left( r \right)}\left( t \right) - \bm{s_i}\left( t \right)} \right\|}^2}} \right)}}{{\left( {{\sigma ^2} + \sum\limits_{j \in {\mathcal K}} {\frac{{{p_j}\left( t \right){\omega _{ij}}}}{{{H^2} + {{\left\| {\bm{q_j}^{\left( r \right)}\left( t \right) - \bm{s_i}\left( t \right)} \right\|}^2}}}} } \right)\ln 2}}} \\
		= D_i^{\left( r \right)}\left( t \right) - \sum\limits_{j \in {\mathcal K}} {E_{ij}^{\left( r \right)}\left( t \right)\left( {{{\left\| {\bm{q_j}\left( t \right) - \bm{s_i}\left( t \right)} \right\|}^2}} \right.} \\
		\left. { - {{\left\| {{\bm q_j}^{\left( r \right)}\left( t \right) - \bm{s_i}\left( t \right)} \right\|}^2}} \right) 
	\end{array}
\end{equation}

Hence, the constraint (\ref{eq:trajectories_slack_problem}b) can be approximated as (\ref{eq:trajectories_problem_constraint_final}), which is a convex constraint.

Additionally, for the introduced non-convex inequality ${B_{ij}}\left( t \right) \le {\left\| {\bm{q_j} \left( t \right) - \bm{s_i}\left( t \right)} \right\|^2}$ , we can calculate the lower bound of its RHS term by 
\begin{equation}\label{eq:slack_variable_sca_proof}
	\begin{array}{l}
		{\left\| {\bm{q_j}\left( t \right) - \bm{s_i}\left( t \right)} \right\|^2} \ge {\left\| {\bm{q_j}^{\left( r \right)}\left( t \right) - \bm{s_i}\left( t \right)} \right\|^2} + \\
		2{\left( {\bm{q_j}^{\left( r \right)}\left( t \right) - \bm{s_i}\left( t \right)} \right)^{\rm{T}}}\left( {\bm{q_j}\left( t \right) - \bm{s_i}\left( t \right)} \right)
	\end{array}
\end{equation}

With (\ref{eq:slack_variable_sca_proof}), we can obtain (\ref{eq:slack_variable_sca}), and it can be observed that (\ref{eq:slack_variable_sca}) is a linear constraint. This completes the proof.

\section{Proof of Lemma \ref{lemma:lemma_trajectories_optimization}}
In this appendix, we discuss how to transform some complex constraints into standard convex cones by introducing some slack variables and prove the equivalence of the transformation.
For the complex constraint (\ref{eq:trajectories_problem_constraint_final}), we introduce a variable ${\zeta _{ij}}\left( t \right)$ to slack the Euclidean norm in it and let
\begin{equation}\label{eq:lemma_trajectory_slack_1}
	\begin{array}{l}
		{\left\| {\bm{q_j}\left( t \right) - \bm{s_i}\left( t \right)} \right\|^2} = {\left( {{x_j^{\left( u \right)}}\left( t \right) - {x_i^{\left( s \right)}}\left( t \right)} \right)^2} \\
		+ {\left( {{y_j^{\left( u \right)}}\left( t \right) - {y_i^{\left( s \right)}}\left( t \right)} \right)^2}
		\le {\zeta _{ij}}\left( t \right),\forall i, j
	\end{array}
\end{equation}

By referring to the similar proof in Appendix B, (\ref{eq:lemma_trajectory_slack_1}) can be rewritten as the following rotated quadratic cone constraint. 
\begin{equation}\label{eq:lemma_trajectory_slack_1_transform_cone}
	\begin{array}{l}
		\left( {{\zeta _{ij}}\left( t \right),\frac{1}{2},{x_j^{\left( u \right)}}\left( t \right) - {x_i^{\left( s \right)}}\left( t \right),{y_j^{\left( u \right)}}\left( t \right) - {y_i^{\left( s \right)}}\left( t \right)} \right) \in Q_r^4,\forall i, j
	\end{array}
\end{equation}

Next, we should prove the equivalence of exploring the slack variable scheme. In particular, we should prove that (\ref{eq:trajectories_problem_constraint_final}) is equivalent to (\ref{eq:lemma_trajectory_slack_1}) and (\ref{eq:new}).
\begin{equation}\label{eq:new}
	\begin{array}{l}
		\sum\limits_{k \in {\mathcal K}} {{c_{ik}}\left( t \right)} \left( {D_i^{\left( r \right)}\left( t \right) - \sum\limits_{j \in {\mathcal K}} {E_{ij}^{\left( r \right)}\left( t \right)\left( {{\zeta _{ij}}\left( t \right) - } \right.} } \right.\\
		\left. {\left. {{{\left\| {\bm{q_j}^{\left( r \right)}\left( t \right) - \bm{s_i}\left( t \right)} \right\|}^2}} \right)} \right) + \sum\limits_{k \in {\mathcal K}} {{c_{ik}}\left( t \right){{\mathord{\tilde 
						\Lambda } }_{ik}}\left( t \right)}  \ge {\eta _i}\left( t \right),\forall i,t
	\end{array}
\end{equation}

For (\ref{eq:trajectories_problem_constraint_final}), if the constraint (\ref{eq:lemma_trajectory_slack_1}) is active, it is not hard to know that the feasible regions generated by (\ref{eq:lemma_trajectory_slack_1}) and (\ref{eq:new}) are the same as (\ref{eq:trajectories_problem_constraint_final}). 
On the contrary, when (\ref{eq:trajectories_problem}) is optimized under the constraints of (\ref{eq:lemma_trajectory_slack_1}) and (\ref{eq:new}), if there is a UAV $j$ or a subscriber $i$ such that (\ref{eq:lemma_trajectory_slack_1}) is non-active, we can always decrease $\zeta_{ij}(t)$ towards $||{\bm q}_j(t) - {\bm s}_i(t)||^2$ without changing the value of (\ref{eq:trajectories_problem}a) and violating (\ref{eq:new}). Therefore, (\ref{eq:trajectories_problem_constraint_final}) is equivalent to (\ref{eq:lemma_trajectory_slack_1}) and (\ref{eq:new}).

Considering that the term ${\mathord{\tilde
		\Lambda } _{ik}}\left( t \right)$ in (\ref{eq:trajectories_problem_constraint_final}) consists of a complex logarithmic function, we introduce a slack variable ${\varphi _{ik}}\left( t \right)$ to handle it and let
\begin{equation}\label{eq:lemma_trajectory_slack_2}
	\begin{array}{l}
		{\log _2}\left( {{\sigma ^2} + \sum\limits_{j \in {\mathcal K}\backslash \left\{ k \right\}} {\frac{{{p_j}\left( t \right){\omega _{ij}}}}{{{H^2} + {B_{ij}}\left( t \right)}}} } \right) \le {\varphi _{ik}}\left( t \right)
	\end{array}
\end{equation}

Further, by introducing the variables $v$ and ${{m_{ij}}\left( t \right)}$ 
and letting their exponential functions
${e^v} = {\sigma ^2}$ and ${e^{{m_{ij}}\left( t \right)}} = \frac{{{p_j}\left( t \right){\omega _{ij}}}}{{{H^2} + {B_{ij}}\left( t \right)}}$,
we can rewrite (\ref{eq:lemma_trajectory_slack_2}) as the following constraint
\begin{equation}\label{eq:lemma_trajectory_slack_2_transform}
	\begin{array}{l}
		{\varphi _{ik}}\left( t \right)\ln 2 \ge \ln \left( {{e^v} + \sum\limits_{j \in {\mathcal K}\backslash \left\{ k \right\}} {{e^{{m_{ij}}\left( t \right)}}} } \right)
	\end{array}
\end{equation}

As the following equivalent transformation holds,
\begin{equation}\label{eq:lemma_standard_exponential_cone}
	\begin{array}{l}
		t \ge \ln \left( {\sum\limits_{j = 1}^n {{e^{{x_j}}}} } \right) \Leftrightarrow \left\{ {\begin{array}{*{20}{c}}
				{\sum\limits_{j = 1}^n {{\mu _j} \le 1} }\\
				{\left( {{\mu _j},1,{x_j} - t} \right) \in {K_{\exp }},\forall j}
		\end{array}} \right.
	\end{array}
\end{equation}
we can transform (\ref{eq:lemma_trajectory_slack_2_transform}) into the following set of convex constraints 
\begin{subequations}\label{eq:lemma_trajectory_slack_2_transform_cone}
	\begin{alignat}{2}
		&\mu  + \sum\limits_{j \in {\mathcal K}\backslash \left\{ k \right\}} {{\mu _{ij}}\left( t \right) \le 1} ,\forall i\\
		&\left\{ {\begin{array}{*{20}{c}}
				{\left( {\mu ,1,v - {\varphi _{ik}}\left( t \right)\ln 2} \right) \in {K_{\exp }} ,\forall i}\\
				{\left( {{\mu _{ij}}\left( t \right),1,{m_{ij}}\left( t \right) - {\varphi _{ik}}\left( t \right)\ln 2} \right) \in {K_{\exp }} ,\forall i, j \ne k}
		\end{array}} \right.
	\end{alignat}
\end{subequations}

Similarly, by introducing the slack variable ${f_{ij}}\left( t \right)$ and letting
\begin{equation}\label{eq:lemma_trajectory_slack_3}
	\begin{array}{l}
		{f_{ij}}\left( t \right) \ge \frac{{{H^2} + {B_{ij}}\left( t \right)}}{{{p_j}\left( t \right){\omega _{ij}}}} = {e^{ - {m_{ij}}\left( t \right)}} 
	\end{array}
\end{equation}
we can obtain the following exponential cone constraint
\begin{equation}\label{eq:lemma_trajectory_slack_3_cone}
	\begin{array}{l}
		\left( {{f_{ij}}\left( t \right),1, - {m_{ij}}\left( t \right)} \right) \in {K_{\exp }} 
	\end{array}
\end{equation}

Similar, the equivalence of introducing the slack variable $f_{ij}(t)$ can be guaranteed.

Finally, with
(\ref{eq:lemma_trajectory_slack_1}) and (\ref{eq:lemma_trajectory_slack_2}), the complex  (\ref{eq:trajectories_problem_constraint_final}) can be transformed into the following linear constraint.
\begin{equation}\label{eq:trajectories_problem_constraint_final_linear}
	\begin{array}{l}
		\sum\limits_{k \in {\mathcal K}} {{c_{ik}}\left( t \right)} \left( {D_i^{\left( r \right)}\left( t \right) - \sum\limits_{j \in {\mathcal K}} {E_{ij}^{\left( r \right)}\left( t \right)\left( {{\zeta _{ij}}\left( t \right) - } \right.} } \right.\\
		\left. {\left. {{{\left\| {q_j^{\left( r \right)}\left( t \right) - \bm{s_i}\left( t \right)} \right\|}^2}} \right)} \right) - \sum\limits_{k \in {\mathcal K}} {{c_{ik}}\left( t \right){\varphi _{ik}}\left( t \right)}  \ge {\eta _i}\left( t \right),\forall i,t
	\end{array}
\end{equation}

According to
(\ref{eq:minimum_safety_distance}) and  (\ref{eq:minimum_safety_distance_sca}), we can obtain the following linear constraint
\begin{equation}\label{eq:minimum_safety_distance_sca_linear}
	\begin{array}{l}
		- {\left\| {\bm {q_k}^{\left( r \right)}\left( t \right) - \bm {q_j}^{\left( r \right)}\left( t \right)} \right\|^2} + 2{\left( {\bm {q_k}^{\left( r \right)}\left( t \right) - \bm {q_j}^{\left( r \right)}\left( t \right)} \right)^{\rm{T}}}\\
		\left( {\bm{q_k}\left( t \right) - \bm{q_j}\left( t \right)} \right) \ge d_{\min }^2,\forall k,k \ne j,t
	\end{array}
\end{equation}

Besides, by referring to the standard expression of a quadratic cone, we can rewrite (\ref{eq:tan_theta}) as
\begin{equation}\label{eq:lemma_tan_theta_cone}
	\begin{array}{l}
		\left( {H{{\tan }^{ - 1}}\theta ,{x_k^{\left( u \right)}}\left( t \right) - {x_i^{\left( s \right)}}\left( t \right),{y_k^{\left( u \right)}}\left( t \right) - {y_i^{\left( s \right)}}\left( t \right)} \right) \\
		\in {Q^3},\forall i,k,t
	\end{array}
\end{equation}
where, ${Q^n} = \left\{ {x \in {R^n}:{x_1} \ge \sqrt {\sum\limits_{j = 2}^n {x_j^2} } } \right\}$ is a quadratic cone of ${{R^n}}$.

Likewise, (\ref{eq:maximum_flight_speed}) can be rewritten as the following quadratic cone constraint  \begin{equation}\label{eq:lemma_maximum_flight_speed_cone}
	\begin{array}{l}
		\left( {{s_{\max }},{x_k^{\left( u \right)}}\left( t \right) - {x_k^{\left( u \right)}}\left( {t - 1} \right),{y_k^{\left( u \right)}}\left( t \right) - {y_k^{\left( u \right)}}\left( {t - 1} \right)} \right) \\
		\in {Q^3},\forall k,t 
	\end{array}
\end{equation}

Based on the above derivation, we can approximately transform (\ref{eq:trajectories_problem}) into the standard conic problem (\ref{eq:trajectories_slack_problem_mosek}). 

{Besides, the lower bounds obtained in (\ref{eq:trajectories_problem_constraint_sca}) and (\ref{eq:slack_variable_sca_proof}) indicate that the feasible region of (\ref{eq:trajectories_slack_problem_mosek}) is a subset of that of (\ref{eq:trajectories_problem}). Therefore, the maximum value of (\ref{eq:trajectories_slack_problem_mosek}a) is the lower bound of the optimal objective value of (\ref{eq:trajectories_problem}).}
This completes the proof. 

\section{Proof of Lemma \ref{lemma:lemma_algorithm_convergence}}
Given a local point (${{\mathcal C}^{(r)}}(t)$, ${{\mathcal P}^{(r)}}(t)$, ${{\mathcal Q}^{(r)}}(t)$) at the ${r}$-th iteration, and denote the corresponding value of (\ref{eq:association_power_trajectories_optimization}a) at this point as $\Phi ({{\mathcal C}^{(r)}}(t),{{\mathcal P}^{(r)}}(t),{{\mathcal Q}^{(r)}}(t))$. By solving (\ref{eq:association_transform_problem}) we can obtain a solution ${{\mathcal C}^{(r+1)}}(t)$ such that $\Phi ({{\mathcal C}^{(r+1)}}(t),{{\mathcal P}^{(r)}}(t),{{\mathcal Q}^{(r)}}(t)) \le \Phi ({{\mathcal C}^{(r)}}(t),{{\mathcal P}^{(r)}}(t),{{\mathcal Q}^{(r)}}(t))$.
Given the local point (${{\mathcal C}^{(r+1)}}(t)$, ${{\mathcal P}^{(r)}}(t)$, ${{\mathcal Q}^{(r)}}(t)$), we can obtain an updated solution ${{\mathcal P}^{(r+1)}}(t)$ by optimizing (\ref{eq:transmission_power_slack_problem_mosek}) and have $\Phi ({{\mathcal C}^{(r+1)}}(t),{{\mathcal P}^{(r+1)}}(t),{{\mathcal Q}^{(r)}}(t)) \le \Phi ({{\mathcal C}^{(r+1)}}(t),{{\mathcal P}^{(r)}}(t),{{\mathcal Q}^{(r)}}(t))$. Similarly, given the updated local point (${{\mathcal C}^{(r+1)}}(t)$, ${{\mathcal P}^{(r+1)}}(t)$, ${{\mathcal Q}^{(r)}}(t)$), we can obtain a new local point (${{\mathcal C}^{(r+1)}}(t)$, ${{\mathcal P}^{(r+1)}}(t)$, ${{\mathcal Q}^{(r+1)}}(t)$) by optimizing (\ref{eq:trajectories_slack_problem_mosek}) at the ${(r+1)}$-th iteration and have $\Phi ({{\mathcal C}^{(r+1)}}(t),{{\mathcal P}^{(r+1)}}(t),{{\mathcal Q}^{(r+1)}}(t)) \le \Phi ({{\mathcal C}^{(r+1)}}(t),{{\mathcal P}^{(r+1)}}(t),{{\mathcal Q}^{(r)}}(t))$.
To this end, we can conclude that $\Phi ({{\mathcal C}^{(r+1)}}(t),{{\mathcal P}^{(r+1)}}(t),{{\mathcal Q}^{(r+1)}}(t)) \le \Phi ({{\mathcal C}^{(r)}}(t),{{\mathcal P}^{(r)}}(t),{{\mathcal Q}^{(r)}}(t))$.
Further, it can be confirmed that (\ref{eq:association_power_trajectories_optimization}a) is low-bounded. Therefore, the iterative optimization Algorithm \ref{alg:alg1} is convergent. This completes the proof.



\ifCLASSOPTIONcaptionsoff
  \newpage
\fi



%

\end{document}